\documentclass[11pt]{article}
\usepackage{fullpage}
\usepackage[noblocks]{authblk}

\usepackage[dvipsnames]{xcolor}

\usepackage{amssymb}
\usepackage{amsthm}
\usepackage{phfqit}

\newtheorem{theorem}{Theorem}[section]
\newtheorem{lemma}[theorem]{Lemma}
\newtheorem{question}[theorem]{Open question}
\theoremstyle{remark}
\newtheorem{remark}[theorem]{Remark}
\theoremstyle{definition}
\newtheorem{definition}[theorem]{Definition}

\newtheorem{proposition}[theorem]{Proposition}

\newtheorem{corollary}[theorem]{Corollary}

\renewcommand{\tr}{\mathrm{Tr}}

\usepackage{hyperref}
\hypersetup{
  colorlinks   = true,    %
  urlcolor     = black,    %
  linkcolor    = black,    %
  citecolor    = black      %
}

\newcommand{\feval}{\mathsf{F}}
\newcommand{\finv}{\mathsf{FInv}}
\newcommand{\secp}{\lambda}

\newcommand{\negl}{\mathsf{negl}}
\newcommand{\binset}{{\{0,1\}}}
\newcommand{\cA}{\mathcal{A}}
\newcommand{\bbN}{\mathbb{N}}    
\newcommand{\bbE}{\mathbb{E}}

\usepackage{todonotes}

\newcommand{\tnote}[1]{\textcolor{blue}{\footnotesize{\bf (Thomas:} {#1}{\bf ) }}}
\newcommand{\rnote}[1]{\textcolor{red}{\footnotesize{\bf (Rotem:} {#1}{\bf ) }}}
\newcommand{\znote}[1]{\textcolor{Violet}{\footnotesize{\bf (Zvika:} {#1}{\bf ) }}}
\renewcommand{\tnote}[1]{}
\renewcommand{\rnote}[1]{}
\renewcommand{\znote}[1]{}

\newcommand{\mH}{\mathcal{H}}
\newcommand{\mA}{\mathcal{A}}

\newcommand{\C}{\mathbb{C}}
\newcommand{\N}{\mathbb{N}}
\newcommand{\eps}{\varepsilon}
\newcommand{\Id}{I}

\newcommand{\mC}{\mathcal{C}}
\newcommand{\lass}{\leq_{\scriptscriptstyle\infty}}
\newcommand{\gass}{\geq_{\scriptscriptstyle\infty}}

\title{Computational Entanglement Theory}

\author{Rotem Arnon-Friedman}
\affil{The Center for Quantum Science and Technology, Faculty of Physics, Weizmann Institute of Science, Israel}
\author{Zvika Brakerski}
\author{Thomas Vidick}
\affil{The Center for Quantum Science and Technology, Faculty of Mathematics and Computer Science, Weizmann Institute of Science, Israel}

\date{\today}

\begin{document}

\maketitle

\begin{abstract}
We initiate a rigorous study of \emph{computational entanglement theory}, inspired by the emerging usefulness of ideas from quantum information theory in computational complexity. 

We define new operational computational measures of entanglement—the computational one-shot entanglement cost and distillable entanglement. We then show that the computational measures are fundamentally different from their information-theoretic counterparts by presenting gaps between them. 
We proceed by refining and extending the definition of \emph{pseudo-entanglement}, introduced in~\cite{aaronson2022quantum}, using the new operational measures; and we present constructions of pseudo-entangled states (for our new definition) based on post-quantum cryptographic assumptions. 

Finally, we discuss the relations between computational entanglement theory and other topics, such as quantum cryptography and notions of pseudoentropy, as well as the relevance of our new definitions to the study of the AdS/CFT correspondence.

We believe that, in addition to the contributions presented in the current manuscript, our work opens multiple research directions, of relevance both to the theoretical quantum information theory community as well as for future applications of quantum networks and cryptography. 
\end{abstract}

\section{Introduction}

    In recent years, many novel insights, in physics as in computer science, have been gained by combining ideas from quantum information theory and from computational complexity. The list is varied and includes, e.g., works in which post-quantum cryptographic assumptions are used to achieve quantum information-theoretic tasks, such as generating certifiable true randomness~\cite{brakerski2021cryptographic}, constructions of pseudo-random states~\cite{ji2018pseudorandom} which lead to  a search for new foundations for quantum cryptography~\cite{morimae2021quantum,AQY21,kretschmer2021quantum}, insightful studies of the complexity of finding and applying unitaries~\cite{bostanci2023unitary}, and leveraging computational complexity arguments to try to address paradoxes at the interface of quantum information and quantum gravity, e.g.~\cite{aaronson2016complexity,bouland2020computational}. 
    Of specific relevance for us are recent works~\cite{gheorghiu2020estimating,aaronson2022quantum}, that introduced the concept of ``pseudo-entanglement'' and applied it to areas ranging from quantum property testing to the study of dualities in quantum gravity.  

    In the information-theoretic setting (i.e., when the computational cost of performing operations is not taken into account) the study of entanglement is termed ``entanglement theory.''
    Entanglement theory is rich and complex: it studies the different forms and features of entanglement, ways to quantify it, tasks that can be performed using entanglement as a resource, operations that increase, preserve or consume it, and much more. 
    Moreover, fundamental results from entanglement theory play a major role in other sub-fields of quantum information theory, such as quantum cryptography and quantum Shannon theory. 

    In this context, the following question arises. 
    \begin{center}
        \emph{How is entanglement theory affected when all operations are required to be computationally efficient?}
    \end{center}
     This question is conceptually natural, and if one is to judge by previous successes at this interface, is likely to lead to important insights on the nature of entanglement and its role in all areas. Moreover, and even though this is not our focus, the question is bound to become increasingly relevant as experimental work on, e.g., quantum networks slowly matures. 
 
    In this work we set the ground for a rigorous study of a resource theory\footnote{See~\cite{chitambar2019quantum} for a general study of resource theories.} that we term \emph{computational entanglement theory}. 
    Towards this, we fix the ``allowed operations'' to be efficient LOCC protocols and define ``computationally maximally entangled states" (more details below).
    We further establish new \emph{computational} analogues of the most widely used entanglement measures, and study their operational meaning and their relation to the information-theoretic measures. 
    In addition, we refine the definition of pseudo-entanglement and discuss implications in quantum cryptography as in quantum gravity. 

    When initiating the study of computational entanglement theory, we build on the wide-ranging study of entanglement in the information-theoretic setting. 
    We present the topics that we discuss in the work and our contributions in what follows. 

    \paragraph{Computational entanglement measures.}
        In entanglement theory, local operations and classical communication (LOCC) maps are used to classify and measure resources, since such operations cannot be used to increase entanglement (they are also called ``free operations'' in the context of resource theories). 
        This is, of course, true independently of whether the maps are efficient or not. 
        Once one takes into account computational considerations, it is natural to use \emph{efficient LOCC maps} and define computational entanglement measures using such maps. We introduce two computational measures: the \emph{one-shot computational entanglement cost} and the \emph{one-shot distillable entanglement}. Similarly to their information-theoretic counterparts, these are defined via optimization problems over computationally efficient LOCC protocols that start (or end) with maximally entangled states as input (or output). The precise definition of the computational measures requires some care; in particular, they are only meaningfully defined for \emph{families} of states, and ``in the limit''-- this is because the widely accepted notion of ``efficient'' in computational complexity refers to polynomial-time computations, and polynomials can always be defined to take arbitrary values at a finite number of points.

         The motivation for considering the computational one-shot measures, as we define them, is that they have an operational meaning (for both pure and mixed states) and are hence of relevance for applications. 
            For example, when considering future quantum networks, entanglement between different nodes in the  network will need to be distilled. 
            When not taking the computational power of the nodes into account, the one-shot distillable entanglement of the states distributed on the network underlies the ability of implementing tasks requiring entanglement over the network (e.g., teleportation or non-local computation).
            Of course, in practice, one will need to take into account the computational power of the nodes in the network. Then, our definitions are the ones that need to be evaluated to determine the network ``capacity'' or other related measures. 
        
				There are some subtleties in the definitions which we point out. In particular, since we are working with LOCC maps, the notion of efficiency departs from the usual notion used when working with a single, ``non-local'', circuit. The states we consider will generally be \emph{bipartite}, meaning that they are states on $n_A+n_B$ qubits, with the first $n_A$ qubits belong to ``system $A$'' and the last $n_B$ belong to ``system $B$.'' Such a bipartition will usually be understood from context; it plays an important role in the semantics attached to the state. In particular, the fact that a bipartite state of low entanglement (implicitly, across the $A:B$ bipartition) can be prepared efficiently using a single circuit \emph{does not} imply that it has low entanglement cost in the LOCC model. This already points to some subtle features of the new theory, some of which we point out, and others will undoubtedly be discovered at later stages.

    \paragraph{Gaps between information theoretic and computational measures.}
    Our first set of results establishes separations between information-theoretic and computational measures of entanglement. These results show that adding the requirement of computational efficiency can drastically increase the entanglement cost of a bipartite state, and drastically reduce its distillable entanglement. In fact, our first result is an existential statement, based on a counting argument, which establishes the following. There exists families of (pure) bipartite quantum states on $n$ qubits for which (i) the information-theoretic entanglement cost is zero but the computational entanglement cost is super-polynomially large, and (ii) the information-theoretic distillable entanglement is maximal, i.e.\ $n$, but the computational distillable entanglement is zero. (See Lemma~\ref{lem:counting-separation} for a precise statement of this result.) 
    
    It is not hard to convince oneself that these separations are optimal. However, one could argue that they have an element of ``cheating'' in them: it is well-known that there exists quantum states that are not efficiently preparable, and so one would expect that the same should apply to states preparable by LOCC (the statement about distillable entanglement is a bit more subtle, and requires a new idea for the proof). 
    To address this limitation of the existential argument, we give two constructions that provide analogous separations, with both constructions  based on families of efficiently preparable states. 

    Firstly, we give a family of efficiently preparable states on $O(n)$ qubits, whose information-theoretic distillable entanglement is $n$ and yet have zero computational distillable entanglement. (See Lemma~\ref{lem:no-comp-dist}.) The definition of the family, as well as this last statement, rely on a standard cryptographic assumption, the existence of secure post-quantum one-way functions. This allows us to define bipartite states that can be constructed efficiently, but whose distillable entanglement is only available to computationally inefficient distillation maps. 
    
Showing a separation for the entanglement cost is more difficult. This is because in order to achieve this we need to devise a family of bipartite states such that the states are efficiently preparable locally, yet they cannot be prepared by any efficient LOCC maps. It is notoriously difficult to study general LOCC maps; for this separation we restrict ourselves to one-way LOCC (meaning that the communication can only take place from a specific side of the bipartition to the other). We also make use of an (arguably) stronger cryptographic assumption, the quantum random oracle (QRO) model (together with an efficient inversion oracle). For any function $\ell(n)\geq n$ we show that there exists a family of states on $2\ell(n)$ qubits that can be efficiently prepared (assuming quantum query access to a classical random oracle and its inverse), and furthermore the states have information-theoretic entanglement cost $n$, yet any efficient one-way LOCC protocol that prepares them must use essentially $\ell(n)$ EPR pairs.  (See Lemma~\ref{lem:ec_gap}).
This bound is optimal, since any efficiently preparable state on $2\ell(n)$ qubits can be prepared efficiently via one-way LOCC and $\ell(n)$ EPR pairs, using teleportation.

    \paragraph{New definition for pseudo-entanglement.}
		
        Having set up the stage with the definition of new entanglement measures we are equipped to formulate a natural definition of pseudo-entanglement. The term ``pseudo-entanglement'' is first introduced in~\cite{aaronson2022quantum} using the von Neumann entropy as a proxy for entanglement, and shown to have a range of potential applications from quantum property testing to quantum gravity. We give a new definition, that we find more natural  and has greater potential for applications where entanglement is used as a resource. We take inspiration from the notion of pseudo-randomness: a distribution is said to be pseudo-random if it is computationally indistinguishable from uniform. We further note that the impact for applications of this definition is that (as long as one only considers polynomial-time processes) it allows one to replace a pseudo-random distribution, which may have a variety of origins, with a uniform distribution, which is much easier to analyze. Finally, one expects that the pseudo-random distribution is a resource which is ``easier'' to produce, in the sense that it requires less randomness than the uniform distribution itself. 
       
			We argue that the right analogue of this definition to the resource of entanglement is the following. Informally, a family of bipartite quantum states $\{\rho_{AB}^k\}$, pure or mixed, each indexed by a \emph{key} $k\in\{0,1\}^*$, is said to be \emph{pseudo-entangled} if two conditions hold. Firstly, it has a low computational entanglement cost ($\rho_{AB}^k$ can be prepared efficiently from the key $k$ and a small number of EPR pairs, using only LOCC operations). Secondly, it cannot be distinguished from a state that is ``maximally entangled'' from the point of view of our theory, i.e.\ the state has maximal entanglement and furthermore this entanglement can be distilled \emph{efficiently} using LOCC. Here, distinguishing should be impossible by any computationally bounded quantum distinguisher, even given a polynomial number of copies of either state. Importantly, while the distinguisher is given polynomially many copies of the state $\rho^k$, it does \emph{not} get the key $k$ (as otherwise it could attempt to distill in order to distinguish low-entanglement from high-entanglement states). This part of the definition is analogous to other definitions in quantum cryptography, such as the definition of pseudorandom states~\cite{ji2018pseudorandom}. See Section~\ref{sec:constructions} for a complete definition.

        For readers familiar with~\cite{aaronson2022quantum}, let us quickly summarize the three main differences between the definitions. 
       (i) We consider mixed states, instead of only pure states. 
       (ii) Our notion of efficiency of the states (not the distinguisher) is in terms of LOCC protocols, instead of allowing general efficient but ``non-local'' circuits.\footnote{Similar notions of efficiency are also considered in the context of relativistic quantum cryptography~\cite{buhrman2014position,speelman2016instantaneous} and quantum interactive games~\cite{gutoski2007toward}.} 
        (iii) We use the one-shot computational entanglement cost and distillable entanglement, which have operational meanings, instead of the entanglement entropy used in~\cite{gheorghiu2020estimating,aaronson2022quantum}.

        We remark that the usage of distillable entanglement instead of the entanglement entropy (defined via the von Neumann entropy) parallels the classical setting: 
        random variables with high \emph{min-entropy} can be distilled to uniform using a randomness extractors. Similarly, quantum states with high \emph{distillable entanglement} can be distilled to maximally entangled states using entanglement distillation protocols. 
        In the same manner that pseudo-randomness is defined using the min-entropy rather than the Shannon entropy, we use the distillable entanglement instead of the entanglement entropy (i.e., the von Neumann entropy of a subsystem).

    \paragraph{Construction of pseudo-entangled states.}
		
        We present constructions for families of pseudo-entangled states that exploit post-quantum cryptographic assumptions. Interestingly, the ``subset-state'' construction of pseudo-entangled states from~\cite{aaronson2022quantum} does \emph{not} seem to fit into our notion of efficiency; while the state are efficiently preparable, the natural LOCC protocol for preparing them uses a maximal number of EPR pairs, thereby effectively ``voiding'' for most applications the fact that these states have low entanglement.\footnote{We note that in~\cite{aaronson2022quantum} the construction allows to discuss the entanglement entropy between any division of the state to two parties. We, on the other hand, fix the division in advance.} (This point is discussed in Section~\ref{sec:seperation}.) To work around this limitation we give somewhat simpler constructions, building on well-known primitives, that do work. Interestingly, our construction seems very close, if not identical (up to a much simpler presentation and analysis) to the construction in the first version of \cite{aaronson2022quantum}. In particular, the construction gives families of states with as low a computational entanglement cost as possible, i.e.\ linear in the security parameter (which can be e.g.\ polylogarithmic in the number of qubits), but are computationally indistinguishable from a family of computational maximally entangled states-- states whose computational distillable entanglement is maximal ($n$). 
							
    \paragraph{Relation to other topics.}
       As mentioned in the beginning of the introduction, entanglement theory is connected to many other fields in quantum information theory and beyond. One can therefore ask how computational entanglement theory may impact the study of other concepts. 
       In Section~\ref{sec:rel-other-top} we discuss two (out of many)  relevant subject areas where there is likely to be an impact. 

       Our first focus is on the connection between pseudo-entanglement, pseudo-randomness and quantum HILL-entropies~\cite{chen2017computational}. 
       In the information theoretic case, due to monogamy of entanglement, there is a clear relation between the two. 
       We discuss the difficulties in the computational case (also related to the recent works regarding the complexity of finding and applying unitaries on a purifying system~\cite{bostanci2023unitary}) and present fundamental open questions.  
    
        The second context in which we consider our result is that of the computational complexity of the AdS/CFT correspondence. 
        Following up on the discussion presented in~\cite{aaronson2022quantum}, we explain why our refined definition of pseudo-entanglement  strengthens the link to quantum gravity by allowing to examine more geometries. Capturing complex geometries is crucial, as it is believed that the AdS/CFT dictionary is computationally hard to compute only in certain cases~\cite{engelhardt2021world}. 
        An important observation is that current research in quantum gravity considers geometries that are linked to mixed quantum states and one-shot information theory~\cite{akers2019large,akers2023one}.
        Our definition of pseudo-entanglement allows to consider such cases and by this expands the relevance of pseudo-entanglement in the context of the study of the complexity of the AdS/CFT correspondence. (For more details see Section~\ref{sec:rel-other-top}).

    \paragraph{Further work and open questions.}
		The main goal of this work is to alert researchers working in quantum information theory, entanglement theory and computational complexity to the possibility for an important, influential computational entanglement theory. We lay elementary foundations for this theory by giving fundamental definitions and pointing out the emergence of subtle features that arise from the unique combination of entanglement and computation; beyond what might be expected from a more superficial assessment. We also open the way for applications of the new theory, to computational information theory, to quantum pseudo-randomness, and to quantum gravity. In these applications we follow the lead of~\cite{aaronson2022quantum}, a work which provided us great inspiration; and yet aim to provide a compelling argument for the possibility of better and more impactful definitions. 
      The main reason that we believe that our definitions open the door for more applications is that~\cite{aaronson2022quantum} are mainly interested in the Schmidt rank of the states and the (in)ability to distinguish states from Haar-random states. While this is termed in~\cite{aaronson2022quantum} pseudo-entanglement, the notions and applications appear to be more closely related to the randomness present in those states rather to the fundamental notion of entanglement. 
			
			Of course we only scratch the surface of the new theory. Other computational measures could be defined. The relation between entanglement manipulation tasks and, e.g., smooth entropies~\cite{tomamichel2015quantum,khatri2020principles} or the computational HILL entropies~\cite{chen2017computational} seem like a particularly fruitful avenue to explore-- do the  connections known in the information-theoretic setting extend to the computational framework? 
			
			In a different direction, it would be interesting to closely relate the notion of pseudo-entanglement to the recent notion of pseudorandom states~\cite{ji2018pseudorandom}, or other fundamental assumptions in quantum cryptography such as ``EFI pairs''~\cite{brakerski2023computational}. Our constructions and previous works establish some connections, and it remains an interesting research direction to further explore such links.

			In terms of applications, the connection between computational entanglement and quantum gravity remains to be explored in detail. As detailed in Section~\ref{sec:ads_cft}, recent works have emphasized the importance of mixed states, and of the use of refined entanglement measures, to develop a fine-grained understanding of dualities in the study of quantum gravity. Separately, the relevance of computational complexity has been repeatedly pointed out, see e.g.~\cite{susskind2016computational,brown2016complexity} and many follow-up works.  Combining the two seems like a \emph{sine qua non}  condition for further progress, which we have to leave to future work.

\section{Entanglement in the information-theoretic setting}\label{sec:em_it}

    This section acts as a preliminary section, devoted to the introduction of entanglement measures in the standard information-theoretic setting. Readers familiar with standard entanglement measures may skip it; no new results are presented. For surveys on entanglement theory we refer to~\cite{plenio2007introduction,horodecki2009quantum}, and for a much more detailed, pedagogical survey of the field we refer to~\cite{khatri2020principles}. For even more general background in quantum information, such as the definition of quantum channels, the fidelity and diamond norm, etc., we refer to the book~\cite{watrous2018theory}.

    A bipartite entangled state is any state that is not separable. That is, it cannot be written as $\psi_{AB}=\sum_i p(i) \psi_A^i \otimes \psi_B^i$, where here the division into $A$ and $B$ components of the state is always implicit from context. 
    To quantify entanglement and compare different states one needs to define \emph{entanglement measures}. We present several entanglement measures and results that are relevant for us. The measures are defined in terms of optimizations over protocols that are based on a combination of (quantum) local operations and two-way classical communication, termed \emph{LOCC protocols}. For convenience we recall the definition.

\begin{definition}\label{def:locc-map} 
    	An LOCC map is any quantum channel 
			\begin{equation}\label{eq:dist_map}
			\Gamma: \mathcal{H}_A \otimes \mathcal{H}_B \to \mathcal{H}_{\bar{A}} \otimes \mathcal{H}_{\bar{B}} 
			\end{equation}
that can be implemented by a two-party interactive protocol, where each party can implement arbitrary local quantum computation and the two parties can exchange arbitrary classical information. Concretely, any such protocol is the sequential iteration of (a) an arbitrary quantum channel on system $A$, followed by a classical message sent from $A$ to $B$, and (b) an arbitrary quantum channel on system $B$, followed by a classical message sent from $B$ to $A$.

Whenever an LOCC map $\Gamma$ as in~\eqref{eq:dist_map} is considered, we will use the notation $n_A = \log_2\dim(\mH_A)$, $n_B= \log_2\dim(\mH_B)$ and $m_A = \log_2\dim(\mH_{\bar{A}})$, $m_B= \log_2\dim(\mH_{\bar{B}})$. If $n_A=n_B$ (resp.\ $m_A=m_B$) then we sometimes also write $n$ for $n_A,n_B$ (resp.\ $m$ for $m_A,m_B$).
\end{definition}
		
        Note that an LOCC map cannot create entanglement; indeed this fact is generally taken as an axiom in the development of any entanglement theory.       
        In certain cases one may restrict the attention to a limited class of, e.g., 1-way LOCC protocols where communication is allowed only in one direction\footnote{1-way LOCC protocols are of special interest since they are the best understood (in the information-theoretic setting). In particular, much more is known regarding the relation between the distillable entanglement, the distillable key using quantum key distribution protocols, and various entropic quantities, when restricting to $1$-way LOCC.}
        or protocols using a small total amount of communication.
        All the definitions in this manuscript can be easily adapted in the obvious manner to restrict the class of maps. 
        We will encounter statements regarding such classes of protocols in Section~\ref{sec:seperation}.

    \subsection{Distillable entanglement}

       Consider an LOCC map $\Gamma$ as in~\eqref{eq:dist_map}, 
       where $\mathcal{H}_A = \mathcal{H}_B = (\C^2)^{\otimes n}$ and $\mathcal{H}_{\bar{A}} = \mathcal{H}_{\bar{B}} = (\C^2)^{\otimes m}$. 
        The maximally entangled state on the space $\mathcal{H}_{\bar{A}} \otimes \mathcal{H}_{\bar{B}}$ is denoted $\Phi^{\otimes m}$, where $\ket{\Phi}=\frac{1}{\sqrt{2}}\left(\ket{00}+\ket{11}\right)$ and $\Phi=\proj{\Phi}$.

        The goal of an entanglement distillation protocol $\Gamma$ is to start with a bipartite state $\rho_{AB}$ and end up with a state close to $\Phi^{\otimes m}$. 
        For $\Gamma$ as in Equation~\eqref{eq:dist_map} we define its \emph{distillation error} on an input state $\rho_{AB}$ as
        \begin{equation}\label{eq:dist_err}
            p_{err}\left(\Gamma, \rho \right) = 1 - F\left(\Gamma(\rho),\Phi^{\otimes m}\right) = 1 - \bra{\Phi^{\otimes m}} \Gamma(\rho) \ket{\Phi^{\otimes m}} \;,
        \end{equation}
where $F$ denotes the fidelity.

	\begin{definition}[One-shot distillable entanglement]
		Let $\varepsilon\in [0,1]$. The one-shot distillable entanglement of a bipartite state $\rho_{AB}\in \mathcal{H}_A \otimes \mathcal{H}_B$ is given by
		\begin{equation}\label{eq:one_shot_E_D}
			E_D^{\varepsilon}(\rho_{AB}) \,=\, \sup_{m, \Gamma} \left\{ \;  m \;  \big|   \; 
 p_{err}\left(\Gamma, \rho \right) \leq \varepsilon  \right\} \;,
				\end{equation}
				where $\Gamma$ is as in Equation~\eqref{eq:dist_map} and $p_{err}$ is given in Equation~\eqref{eq:dist_err}.
		\end{definition}
        Note that in the preceding definition the optimal $\Gamma$ in general depends on the state $\rho$ and on the parameter~$\varepsilon$. 
 
			\begin{definition}[Asymptotic IID distillable entanglement]
				The asymptotic IID distillable entanglement of a bipartite state $\rho\in\mathcal{H}_A \otimes \mathcal{H}_B$ is given by
				\begin{equation}\label{eq:asym_iid_E_D}
					E_D^{\infty}(\rho) = \inf_{\varepsilon\in(0,1]} \,\liminf_{t\rightarrow\infty} \, \tfrac{1}{t} E_{D}^{\varepsilon}\left(\rho_{AB}^{\otimes t} \right) \;.
				\end{equation}
			\end{definition}
	
		We remark that in the preceding definitions $E_{D}^{\varepsilon}$ describes the \emph{number} of approximate maximally entangled states that can be extracted from the state while $E_D^{\infty}$ is the \emph{rate} at which one can extract the entanglement. 
        Different works in the literature use similar notions or consider a ``rate version'' of $E_{D}^{\varepsilon}$, depending on the context. For our work the definitions above suffice.  
			
	The following lemma mentions well-known results that relate the distillable entanglement to entropic quantities. We include them for later convenience.
			
			\begin{lemma}\label{lem:cond_ent}
				For any bipartite state $\rho_{AB}$ and any $\varepsilon\in[0,1]$, 
				\begin{enumerate}
					\item $E_D^{\infty}(\rho) \geq -H(A|B)_{\rho}$,
					\item $E_D^{\varepsilon}(\rho) \geq -H_{\max}^{\varepsilon}(A|B)_{\rho}$.
			\end{enumerate}	
			\end{lemma}
        The first item of the lemma was first proven in~\cite{devetak2005distillation} and the second follows from~\cite{berta2009single}. In both cases, the bounds arise by considering LOCC maps that only use one-way communication. 
       
		\subsection{Entanglement cost}

            Similarly to entanglement distillation, one can consider the ``opposite'' task, of entanglement ``dilution''. The goal is now to use as few maximally entangled states as possible in order to (approximately) create a given state using local operations and classical communication.
            The measure associated with this operational task is termed  \emph{entanglement cost}.

            Define the \emph{dilution error}, similarly to the distillation error in Equation~\eqref{eq:dist_err}, as
            \begin{equation}\label{eq:dilu_err}
            \bar{p}_{err}\left(\Gamma, \rho \right) = 1 - F\left(\Gamma(\Phi^{\otimes n}),\rho\right) \;.
        \end{equation}
			
   \begin{definition}[One-shot  entanglement cost]
		Let $\varepsilon\in [0,1]$. The one-shot  entanglement cost of a bipartite state $\rho_{AB}\in \mathcal{H}_A \otimes \mathcal{H}_B$ is given by
		\begin{equation}\label{eq:one_shot_E_C}
			E_C^{\varepsilon}(\rho_{AB}) = \inf_{n, \Gamma} \left\{ \; n  \;  \big|   \; 
 \bar{p}_{err}\left(\Gamma, \rho \right) \leq \varepsilon  \right\} \;,
				\end{equation}
				where $\Gamma$ is as in Equation~\eqref{eq:dist_map} and $\bar{p}_{err}$ is given in Equation~\eqref{eq:dilu_err}.
		\end{definition}

   \begin{definition}[Asymptotic IID entanglement cost]
				The asymptotic IID entanglement cost of a bipartite state $\rho\in\mathcal{H}_A \otimes \mathcal{H}_B$ is given by
				\begin{equation}\label{eq:asym_iid_E_D}
					E_C^{\infty}(\rho) = \inf_{\varepsilon\in(0,1]} \limsup_{t\rightarrow\infty}  \tfrac{1}{t} E_{C}^{\varepsilon}\left(\rho_{AB}^{\otimes t} \right) \;.
				\end{equation}
			\end{definition}

	\subsection{Elementary relation between entanglement measures}

   The entanglement cost and distillable entanglement introduced in the previous sections are, in some sense, \emph{extremal}, as they bound any other entanglement measure~\cite{horodecki2000limits,plenio2007introduction}. Another widely used entanglement measure 
is the entanglement entropy.

			\begin{definition}[Entanglement entropy]
				The entanglement entropy $E(\rho)$ of a bipartite state $\rho_{AB}$ is given by $E(\rho)=\max\{H(A)_{\rho},H(B)_\rho\}$.\footnote{Here, $H(A)_\rho$ denotes the von Neumann entropy of the reduced density matrix $\rho_A$. For the case of a pure state $\rho_{AB}$, $H(A)_\rho=H(B)_\rho$.}
			\end{definition}

The next lemma clarifies the relation between the three measures introduced so far. 

			\begin{lemma}\label{lem:etremal_em}
				For all bipartite states, 
				\begin{equation}\label{eq:pure_em_ineq}
					E_C^{\infty}(\rho) \geq E(\rho) \geq E_D^{\infty}(\rho) \;.
				\end{equation}
				Furthermore, $E_C^{\infty}$ and $ E_D^{\infty}$ are the extremal entanglement measures. 
				If $\rho_{AB}=\proj{\psi}_{AB}$ is a \emph{pure} state then all entanglement measures are equal and we have 
				\begin{equation}\label{eq:pure_em_collaps}
					E_C^{\infty}(\rho) = E(\rho) = E_D^{\infty}(\rho) \;.
				\end{equation}
			\end{lemma}
			
			Lemma~\ref{lem:etremal_em} and Equation~\eqref{eq:pure_em_collaps} in particular give the entanglement entropy its meaning as an entanglement measure. We emphasize, however, that Equation~\eqref{eq:pure_em_collaps} only holds for \emph{pure states}. For the general case of mixed states, the entanglement entropy is not a unique entanglement measure, and both inequalities in Equation~\eqref{eq:pure_em_ineq} can be very far from tight. Because of this, one has to specify which measure of entanglement one considers. The entanglement cost and the distillable entanglement are measures with the most clear operational interpretation and are generally preferred when studying information processing tasks.

\section{Computational entanglement}
\label{sec:comp-ent}

In this section we give definitions of operational entanglement measures in the computational setting. 

As a first step in Section~\ref{sec:bp_comp} we formalize (following standard conventions) the notions of efficient quantum channels, and efficient quantum LOCC maps. In the context of resource theories, these can be interpreted as the ``free operations'' in our computational entanglement theory. 

Next in Section~\ref{sec:comp_ent_meas} we introduce computational analogues of the most fundamental operationally relevant entanglement measures, the distillable entanglement and the entanglement cost. 

Finally in Section~\ref{sec:comp-me} we give a definition of \emph{computationally maximally-entangled states}. These are a sub-class of maximally-entangled states that are in some sense the ``ideal resource'' in  computational entanglement theory. 

    \subsection{Bipartite computational efficiency}\label{sec:bp_comp}

As in the information-theoretic case, described in Section~\ref{sec:em_it} above, we consider LOCC protocols that respect the bipartition of the quantum states. Now, however, the LOCC protocols, or maps, need to be described by efficient quantum circuits. To make sense of this notion, we switch from talking about a single LOCC map $\Gamma$ to a \emph{family} of maps $\{\hat{\Gamma}^{\lambda}\}_{\lambda}$, where $\lambda\in\mathbb{N}_+$ is an integer that, intuitively, parameterizes the complexity of the map (i.e.\ the map will be required to operate on $\poly(\lambda)$ qubits, and perform at most $\poly(\lambda)$ elementary operations).

 Informally, a sequence $\{\Gamma_{\secp}\}_{\secp \in \bbN_+}$ of LOCC channels is \emph{efficient} if there exists a polynomial $c$ such that for all $\secp$, there is a pair of quantum interactive algorithms, each described by quantum circuits of total size at most $c(\secp)$,\footnote{When implemented over some fixed universal gate set, e.g.\ $\{H,CNOT\}$, together with the creation of ancilla qubits in state $\ket{0}$ and measurements in the computational basis.} that implement the map $\Gamma_\secp$. More formally, we give the following (standard) formalization of computational efficiency for such maps. 

\begin{definition}\label{def:efflocc}
   Given an LOCC map $\Gamma$ (Definition~\ref{def:locc-map}), a \emph{circuit description} of $\Gamma$ is given by two families of circuits $\{\mC_{A,i}\}_{i\in\{1,\ldots,r\}}$ and $\{\mC_{B,i}\}_{i\in\{1,\ldots,r\}}$ each acting on $n_A+t_A+c$ and $n_B+t_B+c$ qubits respectively, such that the following procedure implements the map $\Gamma$ on an arbitrary input $\varphi_{AB} \in (\C^2)^{\otimes n_A} \otimes (\C^2)^{\otimes n_B}$: 
				\begin{enumerate}
				\item Registers $A$ and $B$ of $n_A$ and $n_B$ qubits respectively are initialized in state $\varphi_{AB}$. Ancilla registers $A'$ and $B'$ of $t_A$ and $t_B$ qubits respectively are initialized in the $\ket{0}$ state. Communication register $C$ of $c$ bits is initialized in the $\ket{0}$ state. 
				\item For $i=1,\ldots,r$: The circuit $\mC_{A,i}$ is applied on registers $A$, $A'$ and $C$. Register $C$ is measured in the computational basis. The circuit $\mC_{B,i}$ is applied on registers $B$, $B'$ and $C$. Register $C$ is measured in the computational basis. 
				\item The contents of the first $m_A$ qubits of registers $(A,A')$ and the first $m_B$ qubits of $(B,B')$ are the final output. 
    \end{enumerate}
		We say that $\Gamma$ has a circuit description \emph{of size $c$} if it has a circuit description whose total number of gates (including ancilla creation and qubit measurement) is at most $c$. Finally, we say that the family $\{\Gamma_{\secp}\}_{\secp \in \bbN}$ is \emph{efficient} if there exists a polynomial~$c$ such that for all $\secp$, $\Gamma_{\secp}$ has a circuit description of size at most $c(\secp)$. 
				\end{definition}

The case where the number of rounds in Definition~\ref{def:efflocc} is $r=1$ corresponds to the case of a $1$-way LOCC protocol (where the $1$-way communication is from $A$ to $B$). These protocols are the best-studied, as they contain most existing protocols and are also the class for which most lower bounds holds (due to the difficulty of analyzing the general class). We will show such a lower bound in Section~\ref{sec:sep-ec}.

    \begin{remark}
    	Definition~\ref{def:efflocc} is a non-uniform computational complexity measure. It is also possible to provide a uniform definition in which there exists a polynomial time algorithm that takes $1^\secp$ as input and outputs an implementation of $\Gamma_{\secp}$. This is done for example in~\cite{gutoski2007toward} in the context of quantum interactive games. Our results generally apply to both settings; we take the non-uniform definition for convenience.
    \end{remark}

    \subsection{Computational entanglement measures}\label{sec:comp_ent_meas}

Having defined computational efficiency for LOCC maps, we turn to the relevant entanglement measures --- computational entanglement cost and distillable entanglement.
        Before we start, let us make the following remarks. 
        \begin{enumerate}
            \item All of our measures are defined for \emph{mixed} states. 
            \item We have chosen to work with one-shot measures, since they have an operational meaning. Computational asymptotic IID measures can be easily defined from the one-shot measures in the same manner as in the information theoretic case (see Section~\ref{sec:em_it}).
            \item The measures are well-defined on states that cannot themselves be efficiently prepared by a quantum circuit. Later, when defining pseudo-entanglement (Section~\ref{sec:constructions}), for example, we will restrict our attention to efficiently preparable states. 
        \end{enumerate}
				
As already mentioned, since computational efficiency is only meaningful when considered as a function of an \emph{(input) size parameter} $\lambda$, to adapt the conventional definitions of distillable entanglement and entanglement cost from Section~\ref{sec:em_it} to the computational setting we need to consider \emph{families} of LOCC maps $\Gamma$, and quantum states $\rho$, that depend on $\lambda\in\mathbb{N}_+$. Whenever needed for clarity, we make this dependence explicit by writing the state as $\rho^\lambda$ and the map as $\Gamma^\lambda$. In general, $\rho^\lambda$ will be a bipartite quantum state on $n_A(\lambda)+n_B(\lambda)$ qubits, and $\Gamma^\lambda$ will denote an LOCC map from $n_A(\lambda)+n_B(\lambda)$ to $m_A(\lambda)+m_B(\lambda)$ qubits.

\subsubsection{Computational (one-shot) distillable entanglement}

We denote the computational one-shot distillable entanglement of a family of states $\{\rho^\lambda\}_{\lambda\in\N_+}$, defined formally below, by $\hat{E}^{\varepsilon}_D(\{\rho^\lambda\})$. 
In all places in the manuscript, the hat symbol indicates that we are talking about a ``computational measure''.
Because we are considering families of states and LOCC maps, it would be natural to define the distillable entanglement $\hat{E}^{\varepsilon}_D(\{\rho^\lambda\})$ by a function $m:\N_+\to\N$  such that $m(\lambda)$ denotes ``the number of EPR pairs that can be efficiently distilled from $\rho^\lambda$''. However, it is important to realize that for any such function, the value of $m(\lambda)$ at any given point $\lambda=\lambda_*$ does not carry any meaning by itself. This is because, for a fixed $\lambda^*$, the optimal map $\Gamma^{\lambda_*}$, that distills $E_D^\eps(\rho^{\lambda_*})$ EPR pairs from it, has (obviously) a certain fixed circuit size $c_*$, which can always be considered to be a ``polynomial in $\lambda^*$'', as long as e.g.\ the constant coefficient of the polynomial is chosen large enough. For this reason it seems impossible to define \emph{the} computational distillable entanglement of a state $\rho^{\lambda_*}$, even when considered as part of a family $\{\rho^\lambda\}$. The only possibility is to define a \emph{lower bound} on the computational distillable entanglement $\hat{E}^{\varepsilon}_D(\{\rho^\lambda\})$, and only attribute meaning to the values that such a lower bound $m:\N_+\to\N$ takes ``for $\lambda$ large enough.'' (See Remark~\ref{rk:lb} below for further discussion of this point.)

We now formally state our definition. Below the definition, we make some remarks and introduce notation that allows us to compare different lower bounds on the computational distillable entanglement.

	\begin{definition}[Computational one-shot distillable entanglement]\label{def:comp-dist-nonunif}
		Let $\varepsilon:\mathbb{N}_+\to [0,1]$. Fix polynomial functions $n_A,n_B:\mathbb{N}_+\to\mathbb{N}_+$. Let $\{\rho^\lambda_{AB}\}_{\lambda\in\mathbb{N}_+}$ be a family of quantum states such that for each $\lambda\geq 1$, $\rho^\lambda$ is a bipartite state on $n_A(\lambda)+n_B(\lambda)$ qubits. 
	We say that a function $m:\mathbb{N}_+\to\mathbb{N}_+$ is a \emph{valid lower bound on the computational distillable entanglement} of the family $\{\rho^\lambda\}$ (for short, \emph{lower bound on $\hat{E}^{\varepsilon}_D(\{\rho^\lambda\})$}, or $\hat{E}^{\varepsilon}_D(\{\rho^\lambda\}) \ge m$) if there exists an efficient LOCC map family $\{\hat{\Gamma}^\lambda\}$ such that	for each $\lambda\geq 1$, $\hat{\Gamma}^\lambda$ outputs a $2m(\lambda)$-qubit state, and
						\begin{equation}\label{eq:one_shot_E_D}
 \forall \lambda\in\mathbb{N}_+ \,,\quad p_{err}\left(\Gamma^{\lambda}, \rho^{\lambda} \right) \leq \varepsilon(\lambda) \;,
\end{equation}
				where the function $p_{err}$ is defined in Equation~\eqref{eq:dist_err}.
				
				We omit the superscript and simply write $\hat{E}_D(\{\rho^\lambda\}) \ge m$ if there exists a negligible function~$\eps$ such that $\hat{E}^{\eps}_D(\{\rho^\lambda\}) \ge m$.\footnote{A \emph{negligible function} $\eps:\N_+\to\mathbb{R}_+$ is a function that goes to zero faster than any inverse polynomial. Formally, for any polynomial $p$, $p(\lambda)\eps(\lambda)\to_{\lambda\to\infty} 0$.}
		\end{definition}

We note that in the definition, each $\hat{\Gamma}^\lambda$ is allowed to depend on $\lambda$, and hence on $\rho^\lambda$. This in particular means that the LOCC map ``knows'' what state it is attempting to distill from, analogously to the information-theoretic definition of distillable entanglement. 

\begin{remark}\label{rk:lb}
We expand on the remark made before the definition, regarding the significance of a certain function $m=m(\lambda)$ being ``a lower bound on $\hat{E}_D^\eps(\{\rho^\lambda\})$.'' The important point to realize is that the values $m(\lambda)$ only provide meaningful information ``for large enough $\lambda$''. In particular, any specific value $m(1)$, $m(10)$, $m(10^6)$, etc., does not carry any information. 

As already mentioned, this is because by definition, for any given $\lambda_*$, there is an LOCC map $\Gamma^{\lambda_*}$ that distills $E_D(\rho^{\lambda_*})$ EPR pairs from it. This map has some circuit size $c_*$. Now consider the function $c(\lambda) = c_*$. This function is a polynomial (because it is a constant function). Hence the family of LOCC maps $\Gamma^\lambda = \Gamma^{\lambda_*}$ for all $\lambda$ is an efficient LOCC family, which can be considered in Definition~\ref{def:comp-dist-nonunif}. As a consequence, the function $m_*$ such that $m_*(\lambda_*)=E_D(\rho^{\lambda_*})$ and $m_*(\lambda)=0$ for $\lambda\neq \lambda_*$ is always a valid lower bound on $\hat{E}_D$!

More generally, it is always possible to show that a function $m$ such that $m(\lambda) = E_D(\rho^\lambda)$ for all~$\lambda$ up to any given target value such as $\lambda=10,10^6$, etc., is a valid lower bound on $\hat{E}_D$. Therefore, when considering lower bounds on $\hat{E}_D$, it is essential to focus on the \emph{asymptotic behavior as $\lambda\to \infty$}: The value of any such lower bound at any given point is not meaningful; it is only the value ``for large enough $\lambda$'' that has meaning. 
\end{remark}

The preceding remark motivates the following definition, which enables us to compare functions ``for all $\lambda$ large enough.''

\newcommand{\geas}{\sideset{^{\scriptscriptstyle\lim}}{}{\operatorname{\geq}}}

\begin{definition}\label{def:ass-comparison}
Let $m,n:\mathbb{N}_+\to\mathbb{N}$ be arbitrary functions. Then we say that \emph{$m$ is asymptotically less than $n$}, and write $m\lass n$, if there exists a $\lambda_*$ (which may depend on the functions $m,n$) such that for all $\lambda\geq \lambda_*$, $m(\lambda)\leq n(\lambda)$. Similarly, we use the terminology \emph{$m$ is asymptotically greater than $n$}, denoted $m\gass n$, to signify that  there exists a $\lambda_*$ such that for all $\lambda\geq \lambda_*$, $m(\lambda)\geq n(\lambda)$.
\end{definition}

 Later on, we will work with families of states that are constructed from classical keys $k\in\{0,1\}^{\kappa(\lambda)}$, with $\kappa:\mathbb{N}_+\to\mathbb{N}_+$ an arbitrary polynomially-bounded length function. In this situation, there may be many states $\rho^k$ associated with a single size parameter $\lambda$: all those with a key of length $\kappa(\lambda)$. Towards this we make use of the following \emph{uniform} definition of computational distillable entanglement, where the key $k$ is explicitly passed as an input to the LOCC map $\Gamma$. For convenience we use the simplified notation $\Gamma(k,\rho^k)$ to mean that $\Gamma$ acts on the bipartite state $\proj{k}_{A'}\otimes \rho_{AB}^k\otimes \proj{k}_{B'}$, where the bipartition is $AA':BB'$.

	\begin{definition}[{Uniform} computational one-shot distillable entanglement]\label{def:comp-dist-unif}
			Let $\varepsilon:\mathbb{N}_+\to [0,1]$. Fix polynomial functions $n_A,n_B,\kappa:\mathbb{N}_+\to\mathbb{N}_+$. Let $\{\rho^{k}\}_{k\in\{0,1\}^{\kappa(\lambda)}}$ be a family of quantum states such that for each $\lambda\geq 1$ and $k\in\{0,1\}^{\kappa(\lambda)}$, $\rho^k$ is a bipartite state on $n_A(\lambda)+n_B(\lambda)$ qubits. 
				We say that a function $m:\mathbb{N}_+\to\mathbb{N}_+$ is a \emph{valid lower bound on the computational distillable entanglement} of the family $\{k,\rho^k\}$ (for short, \emph{lower bound on $\hat{E}^{\varepsilon}_D(\{k,\rho^k\})$} or $\hat{E}^{\varepsilon}_D(\{k,\rho^k\}) \ge m$) if there exists an efficient LOCC map family $\{\hat{\Gamma}^\lambda\}$ such that	for each $\lambda\geq 1$, $\hat{\Gamma}^\lambda$ outputs a $2m(\lambda)$-qubit state, and
				\begin{equation}\label{eq:one_shot_E_D_unif}
 \forall \lambda\in\mathbb{N}_+ \,,\forall k\in\{0,1\}^{\kappa(\lambda)}\,,\quad p_{err}\left(\hat{\Gamma}^{\lambda},k, \rho^{k}_{AB} \right) \leq \varepsilon(\lambda) \;,
\end{equation}
			 where
                \begin{equation}\label{eq:dist_err_key}
                   p_{err}\left(\Gamma, k, \rho \right) = 1 - F\left(\Gamma(k,\rho),\Phi^{\otimes m}\right) = 1 - \bra{\Phi^{\otimes m}} \Gamma(k,\rho) \ket{\Phi^{\otimes m}} \;.\footnote{This is  a direct extension of the definition of $p_{err}$ given in Equation~\eqref{eq:dist_err} to the setting in which the distillation map also gets the key as input.}
             \end{equation}             
             We omit the superscript and simply write $\hat{E}_D(\{k,\rho^k\}) \ge m$ if there exists a negligible function $\eps$ such that $\hat{E}^{\varepsilon}_D(\{k,\rho^k\}) \ge m$.
		\end{definition}	
		
		We use the notation $\hat{E}_D^\varepsilon$ for both definitions, Definition~\ref{def:comp-dist-nonunif} and Definition~\ref{def:comp-dist-unif}. It will always be clear from context which definition applies (indeed, this is clear from the form of the input---whether the family is written explicitly or not). We refer to the quantity in Definition~\ref{def:comp-dist-unif} as ``uniform'' because the LOCC map is required to successfully distill from states associated with \emph{all possible} keys (given the key as input). 
		
		In fact, the measure in Definition~\ref{def:comp-dist-nonunif} is easily seen to be a special case of the measure in Definition~\ref{def:comp-dist-unif}; but we separate them for clarity. A setting in which the difference between the two definitions is blurred is when we have a family of states $\{k,\rho_{AB}^k\}$ in mind, but want to consider the computational distillable entanglement associated with a \emph{single} state $\rho^k$ in the family (for each~$\lambda$). In this case, we fix a function $k=k(\lambda)\in \{0,1\}^{\kappa(\lambda)}$ and consider that $m:\mathbb{N}_+\to\mathbb{N}_+$ is a {valid lower bound on the computational distillable entanglement} of the family $\{k,\rho^k\}$ if there exists an efficient LOCC map family $\{\hat{\Gamma}^\lambda\}$ such that	for each $\lambda\geq 1$, $\hat{\Gamma}^\lambda$ outputs a $2m(\lambda)$-qubit state, and
				\begin{equation}\label{eq:one_shot_E_D_bis}
 \forall \lambda\in\mathbb{N}_+ \,,\quad p_{err}\left(\hat{\Gamma}^{\lambda}, \lambda,\rho^{k(\lambda)}_{AB} \right) \leq \varepsilon(\lambda) \;.
\end{equation}
Because in Equation~\eqref{eq:one_shot_E_D_bis} the LOCC family $\hat{\Gamma}$ is only required to distill from a single state $\rho^{k(\lambda)}$, in general a lower bound $m$ on $\hat{E}_D^\eps(\{k,\rho_{AB}^k\}_k)$ is also a lower bound on $\hat{E}_D^\eps(k,\rho_{AB}^k)$, but the converse may not hold. 

\subsubsection{Computational (one-shot) entanglement cost}
In the previous section we discussed the distillable entanglement, which is defined via entanglement distillation protocols. 
We now move to the ``complementary'' task of entanglement dilution --- we start with maximally entangled states and want to create a bipartite state $\rho$ using LOCC protocols. The entanglement cost of $\rho$ is defined via the smallest number of initial EPR pairs needed to create it using such dilution protocols. 
We now consider the computational one-shot entanglement cost, that we denote by $\hat{E}^{\varepsilon}_C$.
Keeping the previous section in mind, $\hat{E}^{\varepsilon}_C$ will be defined via functions $n$ that act as \emph{upper bounds} on $\hat{E}^{\varepsilon}_C$ (while for $\hat{E}^{\varepsilon}_D$ we looked at lower bounds $m$).

The computational one-shot entanglement cost, and its uniform variant, are defined as follows. We first give  the definition for single-state families. 

	\begin{definition}[Computational one-shot entanglement cost]\label{def:comp-cost}
		Let $\varepsilon:\mathbb{N}_+\to [0,1]$. Fix polynomial functions $n_A,n_B:\mathbb{N}_+\to\mathbb{N}_+$. Let $\{\rho^\lambda_{AB}\}_{\lambda\in\mathbb{N}_+}$ be a family of quantum states such that for each $\lambda\geq 1$, $\rho^\lambda$ is a bipartite state on $m_A(\lambda)+m_B(\lambda)$ qubits.  
		We say that a function $n:\mathbb{N}_+\to\mathbb{N}_+$ is a \emph{valid upper bound on the computational  entanglement cost} of the family $\{\rho^\lambda\}$ (for short, \emph{upper bound on $\hat{E}^{\varepsilon}_C(\{\rho^\lambda\})$}, or simply $\hat{E}^{\varepsilon}_C(\{\rho^\lambda\}) \le n$) if there exists an efficient LOCC map family $\{\hat{\Gamma}^\lambda\}$ such that	for each $\lambda\geq 1$, $\hat{\Gamma}^\lambda$ takes as input $n(\lambda)$ EPR pairs, and
			\begin{equation}\label{eq:one_shot_E_C}
 \forall \lambda\in\mathbb{N}_+ \,,\quad \bar{p}_{err}\left(\hat{\Gamma}^{\lambda}, \rho^{\lambda}_{AB} \right) \leq \varepsilon(\lambda) \;,
				\end{equation}
				where 
                \begin{equation}\label{eq:dilu_err_key}
                   \bar{p}_{err}\left(\Gamma, \rho \right) = 1 - F\left(\Gamma(\Phi^{\otimes n}), \rho \right)\;.
             \end{equation}
				The superscript in $\hat{E}_C^\eps$ is omitted when $\eps$ is a negligible function of $\lambda$.
	\end{definition}

Next we give the general definition. 

  \begin{definition}[Uniform computational one-shot entanglement cost]\label{def:comp-cost-unif}
		Let $\varepsilon:\mathbb{N}_+\to [0,1]$. Fix polynomial functions $n_A,n_B,\kappa:\mathbb{N}_+\to\mathbb{N}_+$. Let $\{\rho^{k}\}_{k\in\{0,1\}^{\kappa(\lambda)}}$ be a family of quantum states such that for each $\lambda\geq 1$ and $k\in\{0,1\}^{\kappa(\lambda)}$, $\rho^k$ is a bipartite state on $m_A(\lambda)+m_B(\lambda)$ qubits. 
			We say that a function $n:\mathbb{N}_+\to\mathbb{N}_+$ is a \emph{valid upper bound on the computational  entanglement cost} of the family $\{k,\rho^k\}$ (for short, \emph{upper bound on $\hat{E}^{\eps}_C(\{k,\rho^k\})$}, or $\hat{E}^{\eps}_C(\{k,\rho^k\}) \le n$) if there exists an efficient LOCC map family $\{\hat{\Gamma}^\lambda\}$ such that	for each $\lambda\geq 1$, $\hat{\Gamma}^\lambda$ takes as input $n(\lambda)$ EPR pairs, and
			\begin{equation}\label{eq:one_shot_E_C_unif}
 \forall \lambda\in\mathbb{N}_+ \,, \forall k\in\{0,1\}^{\kappa(\lambda)}\;,\quad \bar{p}_{err}\left(\hat{\Gamma}^{\lambda},k, \rho^{k}_{AB} \right) \leq \varepsilon(\lambda) \;,
				\end{equation}
		where
		 \begin{equation}\label{eq:dilu_err_key_unif}
                   \bar{p}_{err}\left(\Gamma, k, \rho \right) = 1 - F\left(\Gamma(k,\Phi^{\otimes n}), \rho \right)\;.
				\end{equation}		
				The superscript in $\hat{E}_C^\eps$ is omitted when $\eps$ is a negligible function of $\lambda$.
	\end{definition}

\subsubsection{Relation between entanglement measures}

Having introduced all the definitions we need, we end by stating simple relations between them. For simplicity we formulate the relations for the case of a family of pure states $\{k, \rho^k\}_{k\in\{0,1\}^\kappa}$, implicitly parametrized by $\lambda\geq 1$. 

Firstly, we consider relations between entanglement measures and their computational analogue. Because efficient LOCC maps are a subclass of LOCC maps, it holds that any valid lower bound $m$ on $\hat{E}_D^\eps(\{k,\rho_{AB}^k\})$ must satisfy that $m(\lambda)\leq E_D^\eps(k,\rho^k)$ for any $\lambda$ and $k\in\{0,1\}^{\kappa(\lambda)}$. Similarly, any valid upper bound $n$ on $\hat{E}_C^\eps(\{k,\rho^k\})$ must satisfy that $n(\lambda)\geq E_C^\eps(k,\rho^k)$. In the other direction, $\hat{E}_D$ is always non-negative, whereas $\hat{E}_C$ can be unbounded, e.g.\ when the family $\{(k,\rho^k)\}$ is not efficiently preparable, even locally. Stated succinctly, and with slight abuse of notation, we have that for any family,
\[ 0 \,\leq\, \hat{E}_D^\eps \,\leq\, E_D^\eps\qquad\text{and}\qquad E_C^\eps\,\leq\, \hat{E}_C^\eps \,\leq\, +\infty\;. \]
As we will see in the next section (Lemma~\ref{lem:counting-separation}), all inequalities can be strict. In particular, there exists families of maximally entangled states, i.e.\ $2n$-qubit states such that $E_D^0=n$, which have (asymptotically) zero \emph{computational} distillable entanglement. (Of course, there also exist families of maximally entangled states that have \emph{maximal} computational distillable entanglement as well --- we refer to these as ``computational'' maximally-entangled states, see Section~\ref{sec:comp-me}.)

Next we consider relations between computational entanglement measures. Since LOCC maps cannot create entanglement, it always holds that $E_C^0(\rho)\geq E_D^0(\rho)$.\footnote{Using continuity bounds for entanglement the inequality extends, with a small additive loss, to the case of $\eps>0$. For simplicity we consider $\eps=0$.} Using the bounds from the previous paragraph it then follows that for any valid lower bound $m$ on $\hat{E}_D^0(\{k,\rho^k\})$ and any valid upper bound $n$ on $\hat{E}_C^0(\{k,\rho^k\})$, $m(\lambda)\leq n(\lambda)$ for all $\lambda$. Succinctly, 
\begin{equation*}
    \hat{E}_D^0 \,\leq\, \hat{E}_C^0 \;,
\end{equation*}  
as intuitively expected.

\subsection{Computational maximally-entangled states}
\label{sec:comp-me}

We revisit the definition of a maximally-entangled state, introducing a notion that fits the computational setting. 

\begin{definition}[Computational maximally-entangled states family]
	Let $n:\N_+\to\N$ be a polynomially-bounded function. A family of $2n$-qubit  bipartite quantum states $\{\phi^k_{AB}\}_k$, indexed by keys $k\in\{0,1\}^{\kappa(\lambda)}$, is said to be \emph{computational maximally-entangled} if the function $n$ is a valid lower bound on the computational distillable entanglement  of the family $\{k,\phi^k\}$ (recall Definition~\ref{def:comp-dist-unif}).%
\end{definition}

We note that in the definition, the fact that $n(\lambda)$ is a lower bound on $\hat{E}_D^\eps(\{k,\rho^k\})$ implies (for~$\eps$ small enough) that the initial $2n$-qubit state $\rho^k$ is maximally entangled by itself (for every $k$), thereby justifying the terminology ``maximally entangled.'' The ``computational'' qualifier refers to the fact that this entanglement is \emph{accessible}, in the sense of it being distillable into EPR pairs using the key $k$. This requirement is the condition that justifies the term ``computational'' maximally-entangled states to refer to such states. 
We remark that we could also require the states $\rho^k$ to be efficiently preparable (given $k$). We do not do this explicitly in the definition; however it will be the case for our constructions. 

The most trivial computational maximally-entangled states family is simply $\{\Phi^{\otimes n}\}$. 
This trivial family of one state will not be useful for us later. 
Note  also that the Haar random family of states is \emph{not} a computational maximally-entangled states family. This is because, even though a Haar-random state on $2n$ qubits is (close to) maximally entangled with high probability~\cite{page1993average,sen1996average}, this entanglement is not accessible to any efficient distillation protocol (this is easily shown by e.g. an adaptation of the counting argument in Lemma~\ref{lem:counting-separation}). 
We present constructions of computational maximally-entangled states later in this section, see Section~\ref{sec:constructions}.

\section{Separations between computational and information-theoretic entanglement measures}\label{sec:seperation}

\newcommand{\enc}{\mathsf{Enc}}
\newcommand{\dec}{\mathsf{Dec}}

In this section we establish separations between computational and non-computational entanglement measures. We start in Section~\ref{sec:sep-exists} with a lemma that shows that the largest possible separations imaginable can indeed take place, even for \emph{pure} states. This argument, however, is non-constructive, and relies on states that are not efficiently preparable. In Section~\ref{sec:sep-crypto} we give an analogous separation for the case of the distillable entanglement, for efficiently preparable states, under standard post-quantum cryptographic assumptions. Finally in Section~\ref{sec:sep-ec} we give a separation for the entanglement cost for efficiently preparable states, in the quantum random oracle model. 

\subsection{Existential separations}
\label{sec:sep-exists}

The following lemma shows that there exist families of pure states $\{\rho^\lambda\}$ that have high (indeed, maximal) distillable entanglement but (asymptotically) zero computation distillable entanglement; and that there exist families of pure states $\{\sigma^\lambda\}$ that have zero entanglement cost (indeed, they are product states) and yet their computational entanglement cost is very large, larger than any polynomial. 

\begin{lemma}\label{lem:counting-separation}
For any small enough $\eps>0$ and non-decreasing $n(\lambda)$ such that $n(\lambda)\to_{\lambda\to\infty}\infty$ there exists a family $\{\rho_{AB}^\lambda\}$ of bipartite {pure} states on $2n(\lambda)$ qubits such that for all $\lambda$, $E_D^{\varepsilon}(\rho^\lambda)  = E(\rho^\lambda) = n(\lambda)$, but for any valid lower bound $m$ on $\hat{E}_D^\eps(\rho^\lambda)$, we have that $m\lass 0$.\footnote{Recall that according to Definition~\ref{def:ass-comparison}, this means that $m(\lambda)=0$ for all large enough $\lambda$.}

Similarly, there exists a family $\{\sigma_{AB}^\lambda\}$ of bipartite pure states such that $E_C(\sigma^\lambda)=0$ for all $\lambda$ and yet for any valid upper bound $n$ on $\hat{E}_C^\eps(\sigma^\lambda)$, it holds that for any polynomially bounded function $q(\lambda)$, we have that $n\gass q$.\footnote{This condition means that $n$ must grow faster than any polynomial, as $\lambda\to\infty$.}
        \end{lemma}
 
        \begin{proof}
				The proof for the entanglement cost follows from a simple counting argument. The argument for distillable entanglement is also by a counting argument, but it is a bit more subtle. We start with the distillable entanglement. 
				
				First we note that for any polynomial function $n(\lambda)$ %
				and $\eta>0$ there are at least $(1/\eta)^{\Omega(2^{2n(\lambda)})}$ maximally entangled pure states on $2n(\lambda)$ qubits whose pairwise fidelity is less than $1-\eta$. This follows by observing that for $n$-qubit unitaries $U$ and $V$, and $\ket{\Phi^n}$ the tensor product of $n$ EPR pairs, the states $(\Id \otimes U)\ket{\Phi^n}$ and $(\Id \otimes V)\ket{\Phi^n}$ are both maximally entangled and yet
				\[ \Re\big( \bra{\Phi^{n}} (\Id \otimes U^\dagger)(\Id \otimes V)\ket{\Phi^n}\big)\,=\, 1-\frac{1}{2} \frac{1}{2^{n}}\|U-V\|_F^2\;,\]
				i.e.\ any two such states have very low fidelity, unless the Frobenius norm $\|U-V\|_F^2$ is very small, where $\|X\|_F^2 = \tr(XX^\dagger)$. We then use that an $\eta$-net on $n(\lambda)$-qubit unitaries for the normalized Frobenius norm has size  $(1/\eta)^{\Omega(2^{2n(\lambda)})}$. 
				
				Concretely, we construct a set $S_\lambda$ iteratively as follows. Initially, $S_\lambda=\emptyset$. Let $U$ be any $n$-qubit unitary such that $\ket{\psi}=(I\otimes U)\ket{\Phi^n}$  has fidelity less than $1-\eta$ with every state in $S_\lambda$. If such an $U$ exists, then add the $2^{2n}$ pure states $\{ (I\otimes \sigma_X(a)\sigma_Z(b) )\ket{\psi}:\; a,b\in\{0,1\}^n\}$ to $S_\lambda$, where we applied all possible one-time pads on the second half of $\ket{\psi}$. Observe that (i) by construction, any pair of states in $S_\lambda$ has fidelity at most $1-\eta$, (ii) $S_\lambda$ has at least $(1/\eta)^{\Omega(2^{2n(\lambda)})}$ states, and (iii) the uniform mixture over all $\ket{\psi}\in S_\lambda$ is the totally mixed state. Moreover, for any $\rho^\lambda\in S_\lambda$, by construction $E(\rho^\lambda)=n(\lambda)$. 

For each $\rho \in S_\lambda$, let $s(\rho)$ be the smallest size of an LOCC map $\Gamma$ that distills one EPR pair from $\rho$, with error $p_{err}(\Gamma,\rho)\leq \eps$. Let $s(\lambda)=\max_{\rho\in S_\lambda} s(\rho)$. We distinguish two cases. 

Firstly, suppose that $s$ grows faster than any polynomial. This means that there exists a sequence of states $\{\rho^\lambda\}$ such that $\rho^\lambda \in S_\lambda$ for all $\lambda$, and there does not exist any polynomial $p$ such that there would exist a family of LOCC maps $\{\Gamma^\lambda\}$ that distill one EPR pair from $\rho^\lambda$, for any $\lambda$ large enough. Indeed, for this take $\rho^\lambda$ a state that maximizes $s(\rho)$ for $\rho\in S_\lambda$. Then this family $\{\rho^\lambda\}$ establishes the claim made in the lemma. 

Secondly, suppose that $s$ is polynomially bounded. The number of LOCC maps from $2n(\lambda)$ qubits to $2$ qubits of size at most $s(\lambda)$ is at most $2^{\poly(s(\lambda))}$ (where $\poly(\lambda)$ denotes some polynomial function of $\lambda$). This is because any such map can be described using a number of bits that is (at most) polynomial in its size. By the pigeonhole principle there is a $\lambda_*$ such that for all $\lambda \geq \lambda_*$, there is a unique LOCC map $\hat{\Gamma}^\lambda$ of size at most $s(\lambda)$ that distills at least one EPR pair from ``almost all'' states in $S_\lambda$. Quantitatively, for each $\lambda\geq \lambda_*$, for all but a fraction $(1/\eta)^{-\Omega(2^{2n(\lambda)})}$ of states in $S_\lambda$, $\hat{\Gamma}^\lambda(\psi)$ has fidelity at least $1-\eps$ with one EPR pair.

Now let $\rho^\lambda$ be the uniform mixture over all $\psi\in S_\lambda$. Then provided $\eta$ is small enough with respect to $\eps$ it still follows that $\hat{\Gamma}^\lambda(\rho^\lambda)$ has fidelity at least $1-2\eps$ with one EPR pair. However, by property (iii) above the uniform mixture over all $\ket{\psi}\in S_\lambda$ is the totally mixed state. In particular, this mixture has entanglement entropy $0$. For $\eps$ small enough, we have obtained a contradiction. 

The ``similarly'' part of the lemma follows even more directly. Any given family of LOCC maps $\{\hat{\Gamma}^\lambda\}$ using a polynomial number $n(\lambda)$ of EPR pairs as input can only generate one given state $\psi^\lambda$ for each $\lambda$. Therefore, as long as the number of states in $S_\lambda$ is larger than the number of LOCC maps, there is a family of states $\rho^\lambda\in S_\lambda$ such that for all large enough $\lambda$, $\rho^\lambda$ is far from the state generated by any efficient $\hat{\Gamma}^\lambda$. We take a collection $S_\lambda$ that contains $(1/\eta)^{\Omega(2^{n(\lambda)})}$ product states, which guarantees that $E^\eps_C(\rho^\lambda)=0$, and yet for any valid upper bound $n$ on $\hat{E}_C^\varepsilon(\rho^\lambda)$, it must hold that $n(\lambda)=0$ for large enough $\lambda$. 
        \end{proof}

The preceding lemma gives arbitrarily large separations between $E_D$ and $\hat{E}_D$, and between $E_C$ and $\hat{E}_C$. However, a limitation of the separation is that it is only existential. In particular, the states $\{\rho^\lambda\}$ achieving the separation are not efficient, in the sense that they do not have efficient circuits (even without putting the LOCC restriction). This is why it is possible for $\hat{E}_C$ to be superpolynomial; indeed, for any efficiently constructible family of states $\{\rho^\lambda\}$ on $2n(\lambda)$ qubits we will have $\hat{E}_C(\lambda;\rho^\lambda)\leq 2n(\lambda)$ (by teleportation).

\subsection{A cryptographic separation for distillable entanglement (explicit)}
\label{sec:sep-crypto}

To obtain a separation between computational and information-theoretic distillable entanglement for an explicit family of states, and in particular one that is efficient to generate, we leverage computational assumptions. Specifically, we use \emph{non-interactive statistically binding commitment schemes}, which can be instantiated from post-quantum one-way functions (or even from milder assumptions, e.g.\ \emph{EFI pairs}, see~\cite{Yan22,brakerski2023computational}).

\begin{definition}[Non-Interactive Statistically Binding Commitments for Quantum Messages (NISBQ)]
	A NISBQ scheme consists of a family of efficient channels $C_{\secp,n}$ over $n$-qubit input registers, that are computable by circuits of complexity $\poly(\secp,n)$, with the following properties. 
	\begin{enumerate}
		\item There exists a (not necessarily efficient) approximate inverse channel $D$ such that the composition $D \circ C$ is within negligible distance, in the diamond norm, from the identity channel over $n$ qubits.\footnote{The diamond norm is a measure of distance between quantum channels that, informally, generalizes the trace norm for quantums states. See e.g.~\cite{watrous2018theory}.}
		
		\item The channel $C^0_{\secp,n}$ that takes an $n$-qubit input, traces out that input, and then applies $C_{\secp, n}$ on input $\ket{0^n}$ is computationally indistinguishable from $C_{\secp, n}$ for all $n = \poly(\secp)$. That is, any efficient quantum algorithm that is given black-box evaluation access to either of these channels cannot distinguish between them except with negligible advantage.
	\end{enumerate}
\end{definition}

\begin{proposition}
	NISBQ exist if there exist post-quantum one-way functions.
\end{proposition}

\begin{proof}
The proof follows from the existence of quantum non-interactive commitments schemes for classical messages under the assumption that post-quantum one-way functions exist (and in fact even potentially weaker assumptions), see \cite{Yan22,BCQ22} and references therein. In particular, let $\{Q_{b,i}(\lambda)\}_{b\in\{0,1\}, i\in\{1,\ldots,2n\}}$ be a family of $2n$ canonical non-interactive quantum bit commitment schemes, as described in~\cite[Definition 5]{Yan22}, such that the schemes are computationally hiding and statistically binding (for some negligible error $\eps$). 

The channel $C_{\lambda,n}$ is defined as follows. $C_{\lambda,n}$ first applies a quantum one-time pad \cite{AmbainisMTW00} to its $n$-qubit input $\rho$, obtaining $\rho_{a,b}=\sigma_X(a)\sigma_Z(b)\rho(\sigma_X(a)\sigma_Z(b))^\dagger$ with $\sigma_X$, $\sigma_Z$ the Pauli matrices. It then commits to the $2n$ (classical) bits of the pad using a non-interactive statistically binding commitment for classical messages, obtaining quantum states $\sigma_{i}= \tr_R (Q_{c_i,i}(\lambda)\proj{0}Q_{c_i,i}(\lambda))$, where $c_i=a_i$ for $i\in \{1,\ldots,n\}$ and $c_i=b_i$ for $i\in\{n+1,\ldots,2n\}$. Finally, $C_{\lambda,n}$ returns the $n$-qubit padded states $\rho_{a,b}$, together with the $2n$ states $\sigma_1,\ldots,\sigma_{2n}$. 

We argue that the two required properties hold. First, the inversion channel $D$ uses that the commitment scheme is statistically binding to recover the pads $a,b$ from the states $\sigma_i$, with $1-\negl(\lambda)$ success probability. It then undoes the one-time pad to recover $\rho$. Second, we can show that the channels $C^0_{\secp,n}$ and $C_{\lambda,n}$ are computationally indistinguishable through a sequence of hybrids. Using the computational hiding property for the commitment schemes, replacing the states $\sigma_i$ by $\sigma'_{i}= \tr_R (Q_{c'_i,i}(\lambda)\proj{0}Q_{c'_i,i}(\lambda))$, where $c'_i\in \{0,1\}$ is chosen uniformly at random and is independent of $a,b$, leads to a channel $C'_{\lambda,n}$ that is computationally indistinguishable from $C_{\lambda,n}$. Now the only difference between $C'_{\lambda,n}$ and $C^0_{\secp,n}$ is that the former returns its one-time padded input, whereas the latter returns a one-time padded $\ket{0^n}$ state. These are statistically indistinguishable. 

\end{proof}

The next lemma shows that there exists a family of states $\{\rho^\lambda\}$ that can be efficiently prepared, is maximally entangled from an information-theoretic point of view, and yet has no meaningful distillable entanglement in the computational setting.

\begin{lemma}\label{lem:no-comp-dist}
	Assume the existence of a NISBQ scheme. Then there exists a family $\{\rho^\lambda\}$ of pure bipartite states on $2n$-qubits such that:
	\begin{enumerate}
		\item $E_D(\rho^{\secp})=n$. That is, information-theoretic distillation of almost all entropy is possible.
		\item If $\hat{E}^{\varepsilon}_D(\rho^{\secp}) \ge m$ then $m(\secp) \cdot (1-\varepsilon(\secp)-2^{-2m}) = \negl(\secp)$, where $\negl(\lambda)$ denotes any negligible function of $\lambda$. Note that this in particular implies that if $1-\varepsilon$ is inverse-polynomially larger than $2^{-2m}$, then $m \lass 0$.
	\end{enumerate}
\end{lemma}

\begin{proof} Let $C_{\secp,n}$ be a NISBQ scheme. 
	We define $\rho^{\secp}$ as follows.
	\begin{enumerate}
		\item Start with a bipartite EPR state: $\frac{1}{2^{n/2}} \sum_{x \in \binset^n} \ket{x}_A \ket{x}_B$.
		\item Apply the channel $C_{\secp, n}$ locally to the $B$ register.
	\end{enumerate}
	By construction, this state is efficiently generated.
	
	The existence of the inverse channel $D$ implies, essentially by definition, distillation of $n$ EPR pairs to within a negligible infidelity, showing item 1 in the lemma. 
	
	To show item 2, consider the state $\sigma^\secp$ defined using the following generation procedure.
	\begin{enumerate}
	\item Start with a bipartite separable state $\frac{1}{2^{n}} \Id_{A} \otimes \proj{0^n}_B$, where $\Id_A$ denotes the identity over an $n$-qubit space. 
	\item Apply the commitment channel $C_{\secp, n}$ locally to the $B$ register.
\end{enumerate}
Note that $\rho^\lambda$ and $\sigma^\lambda$ can be seen as being prepared using the same procedure, except that the former uses $C_{\lambda,n}$ while the latter uses $C^0_{\lambda,n}$. It follows from the second property of the NISBQ scheme that $\rho^\secp$ and $\sigma^\secp$ are computationally indistinguishable.

Now assume towards a contradiction that $m, \varepsilon$ and an efficient distillation procedure exist, and consider the following candidate distinguisher between $\rho^\secp$ and $\sigma^\secp$. Given an input, the distinguisher applies the efficient distillation procedure on its input state, then projects the outcome on the state containing $m$ EPR pairs and accepts if the projection succeeds.

On one hand, from our assumption on $m, \varepsilon$ and the distillation procedure, there are infinitely many values of $\secp$ such that $m(\secp) > 0$ and $1-\varepsilon- 2^{-2m} \ge 1/p(\secp)$ for some polynomial $p$. For these values of $\secp$, the above process accepts $\rho^\secp$ with probability at least $1-\varepsilon$. On the other hand, since $\sigma^\secp$ is separable, the success probability on it cannot be larger than $2^{-2m}$. It follows that our distinguisher has advantage $1/p(\secp)$ for some polynomial $p$, in contradiction to the states being computationally indistinguishable.
\end{proof}

        An interesting question is whether there are other natural examples when we see a gap between the information-theoretic distillable entanglement and the computational one.
        Our arguments are based either on the hardness of applying some unitaries to ``start'' with inefficient states in a very explicit way, or build on post-quantum cryptographic assumptions. 
      
				\begin{question}
            Are there efficient families $\{\rho^\lambda\}$ such that $E_D^{\varepsilon}(\rho^\lambda) >  \hat{E}_D^{\varepsilon}(\rho^\lambda)$ without invoking post-quantum cryptographic assumptions?
        \end{question}

\subsection{A separations for entanglement cost in the random oracle model (explicit)}
\label{sec:sep-ec}

Showing a separation between computational and information-theoretic entanglement cost, for a family of efficient states, is more difficult. The reason is that one has to argue that even though the states in question can be prepared by efficient circuits, they cannot be prepared by efficient LOCC channels. Proving lower bounds on the existence of LOCC protocols is a notoriously difficult task. 

We conjecture that the family of states from~\cite{aaronson2022quantum} exhibits such a separation (once one fixes an arbitrary cut, and in the appropriate range of parameters). Short of proving this, we give a candidate separation by studying a similar construction, in the random oracle model~\cite{boneh2011random}. We show that even though the states have low information-theoretic entanglement cost (even for \emph{local} protocols, with no communication at all), they have high computational entanglement cost for \emph{one-way} LOCC protocols. Going beyond a lower bound for one-way protocols seems challenging. 

Because the random oracle model is not the focus of the paper, we do not give complete definitions for the model, and only sketch the proof of the lemma below. The lemma and its proof sketch should be understandable to readers familiar with the model.

\begin{lemma}\label{lem:ec_gap} 
Let $\ell:\N_+\to \N$ be a function and $H:\{0,1\}^* \to \{0,1\}^*$ a random oracle mapping strings of length $\lambda$ to strings of length $2\ell(\lambda)$. Let $H^{-1}$ denote the inversion oracle, i.e.\ such that for every $x$, $H^{-1}(H(x))=x$ (if some $y$ has more than one possible preimage, $H^{-1}(y)$ may return any valid preimage). For any integer $\lambda\geq 1$, let 
\[ \ket{\psi^{H,\lambda}} \,=\, \frac{1}{\sqrt{2^\lambda}} \sum_{x\in\{0,1\}^\lambda} \ket{H(x)_L}_A \ket{H(x)_R}_B\;,\]
where we write $H(x)_L$ for the first $\ell(\lambda)$ bits of $H(x)$ (the ``Leftmost'' bits) and $H(x)_R$ for the last $\ell(\lambda)$ bits of $H(x)$. 

Then $E_C^{0}(\psi^{H,\lambda})=\lambda$. Moreover, the state $\psi^{H,\lambda}$ can be prepared efficiently locally using oracle access to both $H$ and $H^{-1}$, and in particular the function $n(\lambda)=\ell(\lambda)$ is a valid upper bound on $\hat{E}_C^0(\psi^{H,\lambda})$. However, for \emph{any} valid upper bound $n'(\lambda)$ on $\hat{E}_C^\eps(\psi^{H,\lambda})$, for small enough $\eps$ and \emph{when restricted to LOCC protocols with one-way communication from $A$ to $B$}, it holds that $n'(\lambda) \gass \ell(\lambda) - \negl(\lambda)$.
\end{lemma}

\begin{proof}
By definition, the von Neumann entropy of $\ket{\psi^{H,n}}$ across the $A:B$ cut is always at most $n$. Moreover, as long as $\ell(n+1)>2n$ then with overwhelming probability over the choice of $H$, the strings $H(x)_L$ are pairwise distinct, for $x\in\{0,1\}^n$, as well as the strings $H(x)_R$. Thus the von Neumann entropy, and hence the entanglement cost, of $\ket{\psi^{H,n}}$ is exactly $n$, with high probability over the choice of $H$. The fact that $\ket{\psi^{H,n}}$ can be prepared efficiently is straightforward, as shown by the following protocol: 
\begin{align*}
\ket{0^{n+\ell(n)}} & \mapsto \frac{1}{\sqrt{2^n}}\sum_{x\in\{0,1\}^n} \ket{x} \ket{0^{\ell(n)}}\\
&\mapsto \frac{1}{\sqrt{2^n}}\sum_{x\in\{0,1\}^n} \ket{x} \ket{H(x)}\\
&\mapsto \frac{1}{\sqrt{2^n}}\sum_{x\in\{0,1\}^n} \ket{0^n} \ket{H(x)}\,=\, \ket{0^n}\otimes \ket{\psi^{H,n}}\;.
\end{align*}
Here, the second line is obtained by performing a query to $H$, and the third line by performing a query to $H^{-1}$. The third line is only valid if $H$ is injective on length-$n$ inputs. As discussed earlier, for $\ell(n)$ large enough this is the case with overwhelming probability. 

We first give the intuition for the lower bound on the computational entanglement cost of the family $\{\psi^{H,n}\}$. A natural protocol for preparing $\ket{\psi^{H,n}}$ starts with $n$ EPR pairs and then evaluates $H$ locally to generate 
\[ \ket{\tilde{\psi}^{H,n}} \,=\, \frac{1}{\sqrt{2^n}} \sum_{x\in\{0,1\}^n} \ket{x}_{A'} \ket{H(x)_L}_A \otimes \ket{H(x)_R}_B \ket{x}_{B'}\;.\]
However, ``erasing'' the string $x$ locally seems to require inverting $H$ given only half of its output, which is easily shown to require an exponential number of queries. (Indeed, this is equivalent to having two independent oracles $H_L(x)=H(x)_L$ and $H_R(x)=H(x)_R$. Access to the inversion oracle for $H$ does not help in inverting either $H_L$ or $H_R$.)

\paragraph{LO protocols.}
We now elaborate on this intuition for exact local protocols (LO), i.e.\ protocols with no communication at all. Suppose for contradiction that there exists a local protocol creating $\ket{\psi^{H,n}}$ from some number $t$ of EPR pairs. A local protocol is specified by two efficient unitaries, $U_A$ and $U_B$, each acting on $t+m$ qubits, and each allowed to make oracle calls to both $H$ and $H^{-1}$ at unit cost.  These unitaries must be such that 
\[ \big(U_A \otimes U_B \big)\Big( \frac{1}{\sqrt{2^t}} \sum_{z\in \{0,1\}^t} \ket{z,0^m}_A \otimes \ket{z,0^m}_B\Big)
\,=\, \frac{1}{\sqrt{2^n}}\sum_{x\in\{0,1\}^n} \ket{H(x)_L}_A \ket{H(x)_R}_B \ket{\Phi}_{A'B'}\;,\]
where $\ket{\Phi}_{A'B'}$ is an arbitrary state of $2(m+t-\ell)$ qubits. Because local unitaries do not change the Schmidt coefficients, $\ket{\Phi}_{A'B'} \propto \sum_{i} \ket{u_i}_{A'}\ket{v_i}_{B'}$ must have uniform Schmidt coefficients. In particular, for every $i$ and $x$ it must be that $U_A^{-1}\ket{H(x)_L}\ket{u_i} = \ket{\varphi_{x,i}}\otimes \ket{0^m}$, for some $t$-qubit state $\ket{\varphi_{x,i}}$. Similarly, it must be that $U_B^{-1}\ket{H(x)_R}\ket{v_i} = \ket{\varphi_{x,i}}\otimes \ket{0^m}$. Crucially, the same state $\ket{\varphi_{x,i}}$ appears on both sides; this is by rotation invariance of the EPR pair. Hence
\[ U_B U_A^{-1}\ket{H(x)_L}\ket{u_i}\,=\, \ket{H(x)_R}\ket{v_i}\;.\]
This gives an efficient procedure, the composition $U_BU_A^{-1}$, that given the state $\ket{u_i}$ (for any $i$), transforms $\ket{H(x)_L}$ into $\ket{H(x)_R}$, for any $x$. The same procedure can, given $H(x)_L$ as a classical input and $\ket{u_i}$ as a quantum ``advice'' state, save a copy of $H(x)_L$, perform the conversion to obtain $H(x)_R$, and recover $x$ by evaluating $H^{-1}$ on $H(x)=H(x)_L\|H(x)_R$. Thus, there is an efficient algorithm that given at most $m+t-\ell$ qubits of advice (which may depend on $H$, but not on $x$), is able to recover $x$ from $H(x)_L$ only.  Such an algorithm is directly ruled out by existing lower bounds for function inversion with advice, e.g.~\cite{chung2020tight}. 

\paragraph{1-way LOCC protocols.}
Finally we give the proof for the case of an exact one-way protocol, where the communication is from $A$ to $B$. A one-way protocol can be decomposed as a sequence of three steps. Firstly, a unitary $U_A$ on $t+m+r$ qubits is applied to $t$ half-EPR pairs, concatenated with $m+r$ ancilla qubits each initialized in the $\ket{0}$ state. Then, the last $r$ qubits are measured in the computational basis to obtain a string $c\in\{0,1\}^r$ that is sent to $B$. Finally, $B$ applies a unitary $U_B^c$ on $t+m$ qubits, containing the second $t$ half-EPR pairs as well as $m$ ancilla qubits each initialized in the $\ket{0}$ state. As in the previous part, for any outcome $c$ obtained in the first step, the result of the transformation should eventually be a state of the form $\ket{\psi^{H,n}} \otimes \ket{\Phi_c}$, for some bipartite state $\ket{\Phi_c}$. However, it is clear that measuring the $A$ part of $\ket{\Phi_c}$ in the Schmidt basis  and sending the outcome to $B$ enables a new protocol that produces $\ket{\psi^{H,n}}\ket{u^{H,n}_c}$, where $\ket{u^{H,n}_c}$ is a pure state on the $B$ side. (This reduction may increase the complexity of the unitary $U_A$, as well as the classical communication, but we will not need this to be bounded.) 

The state after measurement of $c$, conditioned on a specific $c$, can be written as 
\[ \frac{1}{\sqrt{2^n}}\sum_{x\in\{0,1\}^n} \ket{H(x)_L} \ket{\psi_{x,c}} \ket{0^m}\;,\]
for some orthonormal family of states $\{\ket{\psi_{x,c}}\}_x$ on $t$ qubits. At the last step, the unitary $U_B^c$ rotates $\ket{\psi_{x,c}}$ to $\ket{H(x)_R}\ket{u^{H,n}_c}$. Because access to both $H$ and $H^{-1}$ is given, the unitary $(U_B^c)^{-1}$ is also efficient, and maps $\ket{H(x)_R}\ket{u^{H,n}_c}$ to $\ket{\psi_{x,c}}$, for any $x\in \{0,1\}^n$. Importantly, $\ket{u^{H,n}_c}$ is \emph{independent} of $x$.

We now argue that this form of ``compression'', from $\ell(n)$ bits to $t$ bits, is impossible if $t< \ell(n)$ (unless exponentially many queries are made). To show this, we show how to construct an algorithm that compresses \emph{any} $\ell(n)$ classical bits into a $t$-qubit state, without loss of information. By Holevo's theorem, this is impossible if $t<\ell(n)$. The algorithm is defined through the following sequence of hybrids.

\medskip

\noindent {\bf Hybrid $H_0$:} The function $H$ is sampled uniformly at random, and a classical string $c$ and quantum state $\ket{u^{H,n}_c}$, which may depend on $H$, are sampled as described above. $x\in\{0,1\}^n$ is sampled uniformly at random, and we let $z=H(x)_R$.  On input $z$, $c$ and $\ket{u^{H,n}_c}$ algorithm $\mA_1$ executes $(U_B^c)^{-1}$ to prepare $\ket{\psi_{x,c}}$ on $t$ qubits. On input $c$, $\ket{u^{H,n}_c}$ and a $t$-qubit state $\ket{\psi_x}$, algorithm $\mA_2$ reverses algorithm $\mA_1$ to recover $z'$. We define the success probability $p_0$ of the pair $(\mA_1,\mA_2)$ as the probability that $z'=z$; clearly $p_0=1$.

\medskip
\noindent {\bf Hybrid $H_1$:} We first sample $z\in\{0,1\}^{\ell(n)}$ uniformly at random. Next, we sample $x\in\{0,1\}^n$ and $H$ at random conditioned on $H(x)_R=x$. The rest of the experiment proceeds as in $H_1$. Clearly, $H_1$ and $H_2$ are statistically indistinguishable and $p_1=p_0$. 

\medskip
\noindent {\bf Hybrid $H_2$:} We sample $z$, $x$ and $H$ as in the previous experiment. Then, we sample $z^*\in\{0,1\}^{\ell(n)}$ uniformly at random and define a new function $H^*$ by $H^*(x)_R=z^*$, and $H^*(x)=H(x)$ at all other points (also, $H^*(x)_L=H(x)_L$). The rest of the experiment proceeds as before. 

Let $Q$ be an upper bound on the number of queries to $H$ and $H^{-1}$ made by $\mA_1$ and $\mA_2$. By applying the one-way-to-hiding lemma (specifically, Theorem 3 in~\cite{ambainis2019quantum}, which allows us to consider joint distributions over $H,c,\ket{u^{H,n}_c},x$), we obtain an algorithm $\mathcal{B}$ that on input $c,\ket{u^{H,n}_c},z$, with query access to $H$ and $H^{-1}$ is able to recover $x$ with probability at least $\frac{1}{4Q^2}|p_2-p_1|^2$. Thus $\mathcal{B}$ solves an instance of the function inversion with advice problem (where the function is $H(\cdot)_R$), which by~\eqref{chung2020tight} requires an exponential number of queries). Assuming the number of queries made by $\mA_1$ and $\mA_2$ is subexponential, we deduce that $|p_1-p_2|$ is subexponential (in $n$). 

\medskip
\noindent {\bf Hybrid $H_3$:} We sample $z\in\{0,1\}^{\ell(n)}$ and $H$ uniformly at random. We do not sample any $x$. The rest of the experiment proceeds as previously. Clearly, $H_2$ and $H_3$ are statistically indistinguishable and $p_3=p_2$. 

\medskip
\noindent {\bf Hybrid $H_4$:} We proceed as in the previous experiment, except that $H$ is sampled from a family of $2Q$-wise independent hash functions.  Using~\cite[Theorem 3.1]{zhandry2015secure}, $p_4=p_3$.  

Now we observe that Hybrid $H_4$ can be reformulated as follows. First we sample $H$ as in $H_4$, and then $c$ and $\ket{u^{H,n}_c}$. This classical and quantum information is given to two parties Alice and Bob, and considered to be shared randomness (this can be considered shared randomness because we may as well give a complete classical description of $\ket{u^{H,n}_c}$ to both parties; since computational complexity is not being accounted for). Next, on on input a uniformly random $z\in\{0,1\}^{\ell(n)}$, Alice generates the $t$-qubit state $\ket{\psi_{z,c}}$ by executing $\mA_1$ (where we wrote $z$, instead of $x$, since the state now only depends on $z$). Alice sends the state to Bob, who executes $\mA_2$ to recover $z'$. Here, Alice and Bob solve exactly the communication problem considered in~\cite[Theorem 1.1]{nayak2006limits}, from which it follows that $t\geq \ell(n)-\log \frac{1}{p_4}$. Altogether, this proves the theorem. 
\end{proof}

    \section{Pseudo-entanglement and cryptographic constructions}\label{sec:constructions}

In this section we (re-)define pseudo-entanglement, using our new computational, operational measures of entanglement. We then give constructions of pseudo-entangled states that satisfy our new definition, assuming basic cryptographic primitives in post-quantum cryptography (that are implied by the existence of post-quantum one-way functions, which is the only assumption we need). 

\subsection{Pseudo-entanglement}\label{sec:comp_max_ps_def}

We define pseudo-entanglement in the computational setting. Intuitively, a family of states is said to be pseudo-entangled if the states in the family \emph{look} like maximally-entangled states to a computationally bounded quantum algorithm, while at the same time they require a small amount of entanglement \emph{to efficiently prepare}.
This is analogous to pseudo-randomness --- a resource that requires less randomness to prepare than the uniform distribution but still cannot be distinguished from uniform by any computationally bounded distinguisher.  

We give two definitions. The first one is more general, and allows us to refer to a family of states as ``having pseudo-entanglement $(c(\lambda),d(\lambda))$,'' for any pair of functions $c,d:\mathbb{N}_+\to\mathbb{N}$, as long as (a) $c(\lambda)$ is an upper bound on its computational entanglement cost, and (b) it is computationally indistinguishable from a family such that $d(\lambda)$ is a lower bound on its computational distillable entanglement. (Of course, the notion is only interesting for functions $(c,d)$ such that $c<d$, but we state it in general.) In the second definition, we restrict to the case where $d(\lambda)=n(\lambda)$ where the states are on $2n(\lambda)$-qubits, thus considering a family of states to be pseudo-entangled if it is computationally indistinguishable from a computational maximally-entangled family (but itself has lower computational entanglement cost). 

\begin{definition}[Pseudo-entanglement, general definition]\label{def:ps-general}
Let $n,\kappa:\N_+\to \N$ be polynomially bounded functions and $\eps:\N_+\to [0,1]$ and $c,d:\N_+\to \N$ be arbitrary. 
	A family of (potentially mixed) $2n(\lambda)$-qubit bipartite quantum states $\{\psi^k_{AB}\}_k$ is said to  \emph{have pseudo-entanglement $(\eps,c,d)$} if there is a family of $2n(\lambda)$-qubit bipartite quantum states $\{\phi^k_{AB}\}_k$ such that the following conditions hold. 
	\begin{enumerate}
		\item The function  $c$ is a valid upper bound on $\hat{E}_C^{\varepsilon}(\{k,\psi^k_{AB}\}_k)$. 
		\item The function $d$ is a valid lower bound on $\hat{E}_D^\eps(\{k,\phi^k_{AB}\}_k)$
		\item For any polynomial $p(\lambda)$, no poly-time quantum algorithm can distinguish between the ensembles $\rho_{\mathbf{AB}}=\mathbb{E}_k\left[{\psi^k_{AB}}^{\otimes p(\lambda)}\right]$ and $\varphi_{\mathbf{AB}}=\mathbb{E}_k\left[{\phi^k_{AB}}^{\otimes p(\lambda)}\right]$ with more than negligible probability.
		That is, for any poly-time quantum algorithm $\mathcal{A}$, 
		\begin{equation}
			\big| \mathcal{A}(\rho_{\mathbf{AB}}) - \mathcal{A}(\varphi_{\mathbf{AB}}) \big| \leq \negl(\lambda) \;.
		\end{equation}
	\end{enumerate}
\end{definition}

This definition can be compared to the definition of pseudo-entanglement in~\cite{aaronson2022quantum}. The main differences are as follows. Our first condition is a stricter requirement than the condition in~\cite{aaronson2022quantum} that $\{k,\psi^k_{AB}\}_k$ is efficiently preparable and has entanglement entropy at most $c$. This is because having low computational entanglement cost implies both. However, it seems more relevant to us because for the definition from~\cite{aaronson2022quantum}, even though the state has low entanglement it may not be possible to prepare it in a local manner using few EPR pairs. Our second condition is incomparable to the condition in~\cite{aaronson2022quantum} that $\{k,\phi^k_{AB}\}_k$ is efficiently preparable and has entanglement entropy at least $d$. On the one hand, we do not require that $\{k,\phi^k_{AB}\}_k$ is efficiently preparable, because this seems irrelevant for applications (though, for our examples below, this condition is met). However, we do require that the high entanglement in $\{k,\phi^k_{AB}\}_k$ is efficiently accessible, i.e.\ we impose the stronger condition that the family has high computational distillable entanglement. 

Finally, the last condition on indistinguishability is identical to the analogous condition in~\cite{aaronson2022quantum}. In particular, in both cases we consider polynomial-time distinguishers that are given a polynomial number of copies of the state. The restriction to polynomial-time distinguishers is standard, and allows one to use cryptographic assumptions to construct the states. (One could restrict to LOCC distinguishers, and obtain a more permissive notion; we decide to keep the stronger one.) Similarly, one could give a single copy of the state to the distinguisher. However, this would lead to a very weak requirement; as indeed it would suffice for two families to lead to the save density, such as a totally mixed state, when averaged over the key, to be considered ``pseudo-'' of each other. However, a moment of thought reveals that this is very weak-- consider e.g.\ a family formed by all four $2$-qubit Bell states, compared to all $2$-qubit computational basis states; these lead to the same average density, but one would not consider them ``indistinguishable'' by any reasonable measure. The use of multiple copies is furthermore explicitly required for some applications, such as for example to property testing or potentially to IID (or at least ``many copies'') distillation. Since, finally, this only makes the definition stronger, we include the general condition. 

To conclude this comparison we note that in a concurrent work~\cite{bouland2023public} the authors introduce an extension of the definition of pseudo-entanglement from~\cite{aaronson2022quantum} where the distinguisher $\mathcal{A}$ is given circuits that create the quantum states, instead of copies of the states themselves. A similar extension can be made to our definition when this is of interest. %
Whichever definition is chosen for the computational distinguishability part of the definition, the crucial aspect for us is the requirement on the computational entanglement measures in items 1.\ and 2.\ of Definition~\ref{def:ps-general}.

Having stated our general definition we specialize it to define pseudo-\emph{maximally}-entangled states. 

\begin{definition}[Pseudo-maximally-entangled states family]\label{def:ps-ent}
Let $n,\kappa:\N_+\to \N$ be polynomially bounded functions and $\eps:\N_+\to [0,1]$ and $c:\N_+\to \N$ be arbitrary.	
	A family of (potentially mixed) $2n(\lambda)$-qubit bipartite quantum states $\{\psi^k_{AB}\}_k$ is said to be \emph{$(\eps,c)$-pseudo-maximally-entangled} if it has pseudo-entanglement $(\eps,c,n)$.
\end{definition}

We note that the preceding definition is equivalent to requiring that there is a family  $\{\phi^k_{AB}\}_k$ of computational maximally-entangled states over $2n(\lambda)$ qubits such that the ensembles $\{ \mathbb{E}_k[{\psi^k_{AB}}^{\otimes p(\lambda)}]\}_\lambda$ and $\{\mathbb{E}_k[{\phi^k_{AB}}^{\otimes p(\lambda)}]\}_\lambda$ are computationally indistinguishable, for any polynomial $p$.

        \subsection{Cryptographic primitives}

Before giving explicit constructions in the next section, we review the cryptographic primitives that we will make use of. 
    
    \begin{definition}[PRP] A Pseudorandom Permutation Family (PRP) is defined by a pair of classical deterministic polynomial time algorithms as follows.
    	\begin{itemize}
    		\item $\feval(k, x)$ %
    		takes as input a key $k \in \binset^n$ and an input $x \in \binset^n$. It returns $y \in \binset^n$.
    		\item $\finv(k, y)$ %
    		takes as input a key $k \in \binset^n$ and an input $y \in \binset^n$. It returns $x \in \binset^n$.
    	\end{itemize}
    	We note that here we chose the key and input lengths to be equal (and implicitly both are also equal to the security parameter $\secp$) for the sake of minimizing the number of parameters.
    	
    	We require the following properties.
    	\begin{itemize}
    		\item \textbf{Correctness.} For all $n$, for all $k\in \binset^n$, and for all $x, y \in \binset^n$, it holds that $\feval(k, x)=y$ if and only if $\finv(k, y)=x$.
    		
    		\item \textbf{Pseudorandomness.} For any polynomial time (quantum) algorithm $\cA$ with binary output it holds that
    		   		\[ \left| \Pr_{k \gets \binset^n}[\cA^{\feval(k, \cdot),\finv(k, \cdot)}(1^n)=1] - \Pr_{\pi \gets S_{\binset^n}}[\cA^{\pi, \pi^{-1}}(1^n)=1] \right| = \negl(n)~.\]
Here, the superscript notation $\cA^{\feval(k, \cdot),\finv(k, \cdot)}$ means that the quantum algorithm $\cA$ is allowed to make arbitrary (quantum) queries to evaluate the functions $\feval(k, \cdot)$ and $\finv(k, \cdot)$.
    	\end{itemize}
    \end{definition}
    The existence of PRP is known to be equivalent to the existence of one-way functions under classical reductions (see \cite{GoldreichBook1}). A construction that is secure against quantum adversaries was first given  by Zhandry \cite{Zhandry16prp}.
    
    \begin{remark}
    	For our purposes, we only require a weaker notion of ``one-directional'' pseudorandomness for our PRP. Namely, indistinguishability from random against an adversary that has access to either $\feval$ or $\pi$, but not to their inverses. This weaker notion may be easier to achieve (e.g.\ under weaker assumptions or with higher efficiency) and can readily be plugged into our result.
    \end{remark}

The following lemma will be used to analyze the constructions in the next section. 

\begin{lemma}\label{lem:shrinkblow}
Consider a PRP family $(\feval,\finv)$ and let $n, m \in \bbN$ be s.t. $m>n$. We consider $k_f, k_g, k_h \in \binset^m$, and define the following functions: $f: \binset^m \to \binset^m$ as $f(x) = \feval(k_f, x)$, $g: \binset^n \to \binset^m$ as $g(x) = \feval(k_g, x \| 0^{m-n})$, and $h: \binset^m \to \binset^n$ as $h(x) = \feval(k_h, x)_{|n}$ (i.e.\ the $n$-bit prefix of the outcome). Then for any polynomial function $m=m(n)$ and for any polynomial-time quantum algorithm $\cA$ with binary output it holds that 
    		\[ \left| \Pr_{k_g, k_h}[\cA^{g(h(\cdot))}(1^n)=1] - \Pr_{k_f}[\cA^{f(\cdot)}(1^n)=1] \right| = \negl(n)~.\]
\end{lemma}
\begin{proof}[Proof Sketch]
	The proof proceeds through a sequence of hybrids.
	Consider an adversary $\cA$ and define the random variable $H_0 = H_0(n) = \cA^{g(h(\cdot))}(1^n)$, where $g,h$ are defined based on uniformly random keys $k_g, k_h$.
	
	We then let $\pi_g, \pi_h$ be uniformly random permutations on $\{0,1\}^m$. In $H_1$ we change $h$ to be defined using $\pi_h$ instead of $\feval(k_h, \cdot)$. It holds that $\norm{H_0-H_1}_1=\negl(n)$ because of PRP pseudorandomness, where $\norm{\cdot}_1$ denotes $L_1$ norm. Similarly, in $H_2$ we also change $g$ to be defined using $\pi_g$ instead of $\feval(k_g, \cdot)$. Again we have $\norm{H_1-H_2}=\negl(n)$ for the same reason.
	
	The next few hybrids do not rely on computational assumptions but rather on query lower bounds to random permutations and functions. Let $Q$ be an upper bound on the number of queries made by $\mA$ to its oracle. The next bounds only rely on the assumption that $Q=\poly(n)$. 

In hybrid $H_3$ we change $\pi_h$ to be a random \emph{function} rather than a permutation. It holds that $\norm{H_2-H_3}=\negl(n)$ using the query lower bound for distinguishing random functions from random permutations \cite{Yuen13}. Hybrid $H_4$ does the same for $\pi_g$, and by the same theorem $\norm{H_3-H_4}=\negl(n)$.
	
	In $H_5$ we make a more radical change. Let $r: \binset^m \to \binset^m$ be a random function, we let $H_5 = \cA^{r(\cdot)}(1^n)$. We show that $\norm{H_4-H_5}=\negl(n)$ using Zhandry's small domain theorem \cite[Theorem~1.1]{Zhandry12}. 
	
	 Fix the function $g(x)=\pi_g(x\|0^{m-n})$. This function defines a multi-set $S_g\subseteq \{0,1\}^m$ of size $2^n$ (namely, it possibly has repetitions). Let $D_{1,g}$ be the uniform distribution over $S_g$. Let $D_{2}$ be the uniform distribution over $\{0,1\}^m$. Using~\cite[Theorem~1.1]{Zhandry12}, an algorithm that distinguishes $H_4$ from $H_5$ with some advantage can be transformed into an algorithm that distinguishes $D_{1,g}$ from $D_2$ by getting polynomially samples from either distribution. Since the reduction is uniform, this distinguisher has advantage in expectation over $g$ as well. However, when $g$ is chosen at random, polynomially many samples from $D_{1,g}$ are statistically indistinguishable from polynomially many samples from $D_2$, which completes the argument. 
	
		Next, we define $H_6 = \cA^{\pi_f(\cdot)}(1^n)$, where $\pi_f$ is a random permutation. Here again we use the function to permutation indistinguishability to argue that $\norm{H_5-H_6}=\negl(n)$. 
	
	Our last hybrid again relies on a cryptographic assumption. In $H_7$ we replace $\pi_f$ with $f$ and use PRP pseudorandomness to deduce that $\norm{H_6-H_7}=\negl(n)$.
	
	We can conclude that $\norm{H_0-H_7}=\negl(n)$ and the lemma follows.
\end{proof}

\subsection{Our construction}
\label{sec:constructions}

\begin{definition}\label{def:comp-pse}

Consider a PRP family $(\feval,\finv)$. For all polynomially bounded functions $n, m:\N_+\to\N$ such that $m>n$ we define the following. Let $\kappa(\lambda)=2m(\lambda)$. For $m=m(\lambda)$ and $k_f, k_g, k_h \in \binset^m$, define  $f(x) = \feval(k_f, x)$, $g(x) = \feval(k_g, x \| 0^{m-n})$, $h(x) = \feval(k_h, x)_{|n}$. Note that $g,h$ are efficiently computable given $k = (k_g, k_h)$. 

We now define our state ensembles. For any $\lambda$ and $k = (k_g, k_h)\in\{0,1\}^{\kappa(\lambda)}$, define
\begin{align}\label{eq:comp-psik}
 \ket{\psi^{k}} \,=\, \frac{1}{\sqrt{2^{m}}} \sum_{x\in \{0,1\}^{m}} \ket{x}_A \otimes \ket{g(h(x))}_B \;,
	\end{align}
and for $k'=k_f\in\{0,1\}^{m(\lambda)}$,
\begin{align} \label{eq:comp-phik}
	\ket{\varphi^{k'}} \,=\, \frac{1}{\sqrt{2^{m}}} \sum_{x\in \{0,1\}^{m}} \ket{x}_A \otimes \ket{f(x)}_B \;.
\end{align}
\end{definition}

It is immediate from the definitions that both families $\{(k,\psi^k)\}$ and $\{(k',\varphi^{k'})\}$ are efficiently computable. The following lemma shows that they are also computationally indistinguishable (of course, this relies on the distinguisher not having access to the key).

\begin{lemma}\label{lem:phipsiind}
	For all $t=\poly(\lambda)$, it holds that $\bbE_k[(\ket{\psi_{k}}\bra{\psi_{k}})^{\otimes t}]$ and  $\bbE_{k'}[(\ket{\varphi_{k'}}\bra{\varphi_{k'}})^{\otimes t}]$ are computationally indistinguishable.
\end{lemma}
\begin{proof}
	The lemma immediately follows from Lemma~\ref{lem:shrinkblow} above.
\end{proof}

We note that these families are very similar to other families of pseudorandom or pseudo-entangled states that have been previously considered in the literature. However, in contrast to other constructions, here we are able to give a good upper bound on the computational entanglement cost.

\begin{lemma}\label{lem:psicost}
	There exists an efficient LOCC family $\{\Gamma^\lambda\}$ that takes $k\in\{0,1\}^{2m(\lambda)}$ and $n(\lambda)$ EPR pairs as input and returns $\psi^{k}$. In other words, $\hat{E}_C^0(k,\psi^k)\lass n$. 
\end{lemma}

\begin{proof}
	We propose the following LOCC (in fact LO) protocol.
	\begin{enumerate}
		\item Start with $\frac{1}{\sqrt{2^n}}\sum_{y \in \binset^n} \ket{y}_{A_0}  \otimes \ket{y}_{B_0}$.
		\item Locally add registers: $\frac{1}{\sqrt{2^n}}\sum_{y \in \binset^n} \ket{y}_{A_0} \ket{+}^{\otimes(m-n)}_{A_1} \otimes \ket{y}_{B_0} \ket{0}^{\otimes(m-n)}_{B_1}$.
		\item Party $A$ locally applies the unitary $\ket{x} \to \ket{\finv(k_h, x)}$ to its local register $A=A_0 A_1$. We get (after change of variable) $\frac{1}{\sqrt{2^m}}\sum_{x \in \binset^m} \ket{x}_{A} \otimes \ket{h(x)}_{B_0} \ket{0}^{\otimes(m-n)}_{B_1}$.
		\item Party $B$ locally applies the unitary $\ket{x} \to \ket{\feval(k_g, x)}$ to its local register $B=B_0 B_1$. We get $\frac{1}{\sqrt{2^m}}\sum_{x \in \binset^m} \ket{x}_{A} \otimes \ket{g(h(x))}_{B}$.
	\end{enumerate}
	The resulting state is exactly $\ket{\psi_{k}}$.
\end{proof}

Finally, and in contrast to Haar-random states, for our construction the highly entangled state family also has high, indeed maximal, computational distillable entanglement. 

\begin{lemma}\label{lem:phidistill}
	There exists an efficient LOCC family $\{\Gamma^\lambda\}$ that takes $k'\in\{0,1\}^{m(\lambda)}$  as input and produces $m$ EPR pairs.	In other words, $\hat{E}_D^0(k,\psi^k)\gass m$.  
\end{lemma}

\begin{proof}
This follows essentially by definition, since party $B$ can locally apply $\finv(k_f, \cdot)$, to obtain the maximally entangled state.
\end{proof}

Combining the previous lemma gives the following. 

\begin{corollary}
	Let $n(\secp), m(\secp)$ be polynomial functions s.t.\ $\secp \le n(\secp) \le m(\secp)$ for all $\secp$.
	Assuming the existence of quantum-secure one-way functions, there exists a family of  (pure) $2m$-qubit states that has pseudo-entanglement $(0,n,m)$. In particular, this family is $(0,n)$-pseudo-maximally-entangled.
\end{corollary}

\begin{proof}
	Let $n,m$ be as in the corollary statement. Consider, for all $\secp$, the states $\rho_{AB}^k = \ket{\psi_{k}}\bra{\psi_{k}}$ and $\sigma_{AB}^{k'}=\ket{\varphi_{k'}}\bra{\varphi_{k'}}$, defined in~\eqref{eq:comp-psik} and~\eqref{eq:comp-phik} respectively. 
	
	Lemma~\ref{lem:phipsiind} guarantees that the two ensembles are computationally indistinguishable for any polynomial number of copies. Lemma~\ref{lem:psicost} shows that $n$ is an upper bound on $\hat{E}_C^{0}(\{k,\psi_{k}\}_k)$, and Lemma~\ref{lem:phidistill} shows that $m$ is a lower bound on $\hat{E}_D^{0}(\{k,\sigma_{k}\}_k)$.
\end{proof}

\section{Relation to other topics}\label{sec:rel-other-top}

    Building on our approach one may continue to develop what we term ``computational entanglement theory.'' Questions that have up until now been investigated in ``standard'' entanglement theory can be asked in the computational context. Perhaps surprisingly, translating well-known results in quantum information theory to the computational setting is not trivial; by which we mean that it does not follow by simply adding the word ``efficient'' where it seems needed. 
    In the next section we discuss a notable and important example-- the relation between entanglement and randomness. Following that, we proceed and examine the connections between our model computational entanglemenet and quantum gravity.

    \subsection{Pseudo-entanglement vs.\@ pseudo-randomness}

    It is well known that entanglement is ``monogamous'': If Alice and Bob are maximally entangled then Eve, a third party, must be completely decoupled from them. Mathematically, for every state $\Psi_{ABE}$ such that $\mathrm{Tr}_E \Psi_{ABE} = \Phi_{AB}$, where $\Phi_{AB}$ is a maximally entangled state, it holds that $\Psi_{ABE}=\Phi_{AB}\otimes \psi_E$. This, in particular, implies that when Alice performs, e.g., an $X$ or $Z$ basis measurement on her system $A$, the outcome is completely \emph{random} from Eve's point of view (i.e.\ even when given access to $\psi_E$). 
    Techniques such as entropic uncertainty relations and decoupling theorems (see for example~\cite{dupuis2014one,coles2017entropic}) allow one to make similar statements quantitative in the case when Alice and Bob's state is not a perfect maximally entangled state.
		
		Arguments such as this one allow us to see entanglement and true randomness as, in some sense, two sides of the same coin. This well-known observation raises an analogous fundamental question regarding the relation between computational entanglement and randomness, or pseudo-entanglement and pseudo-randomness. 
    To phrase it emphatically, \emph{is computational entanglement monogamous}, in some quantitatively meaningful way? 
    
    Let us first consider the information-theoretic setting.
    For a pure state~$\Psi_{ABE}$, given a lower bound on the (quantum) conditional min-entropy $H_{\min}(A|E)$, one can decouple Alice's state from Eve's by preforming only local operations on Alice's system; this is shown for example in~\cite{dupuis2014one}. Denote the outcome of the process by $\tilde{\psi}_{\tilde{A}E}=\mathbb{I}_{\tilde{A}}\otimes \psi_{E}$, where~$\mathbb{I}_{\tilde{A}}$ is the maximally mixed state on a smaller system $\tilde{A}$.\footnote{The setup here is a quantum version of classical randomness extractors; In the classical case, by applying an extractor on a classical random variable $A$ we end with a smaller system $\tilde{A}$ which is close to uniform even given the side information, assuming the conditional min-entropy is sufficiently high. See also~\cite{berta2013quantum}.} If Alice was to measure her state, she would get true randomness.
    Now, since Alice and Eve's state is $\mathbb{I}_{\tilde{A}}\otimes \psi_{E}$, Uhlmann's theorem implies that there exists a unitary on Bob's system~$B$ that will ``locate'' a subspace $\tilde{B}$ such that the final state on the registers $\tilde{A}\tilde{B}E$ is $\tilde{\Psi}_{\tilde{A}\tilde{B}E}=\Phi_{\tilde{A}\tilde{B}}\otimes \psi_{E}$.
    In other words, we observe that distilling entanglement between Alice and Bob and decoupling Alice's system from Eve have the same effect. 
    
    The statement that we just gave, however, is based on the \emph{existence} of a unitary on Bob's side. Unfortunately, the fact that the unitary exists does not mean that finding and applying it is tractable. In fact, it implicitly follows from~\cite{aaronson2022quantum} that the problem is hard in general. The complexity of finding and applying the unitary is further studied in the recent work~\cite{bostanci2023unitary}. 
    Thus, in the computational case we cannot simply use Uhlmann's theorem as it may not lead to an efficient distillation protocol. 

The question of relating computational measures of entanglement and of randomness can also be approached from the angle of entropic measures. 
    The one-shot distillable entanglement $E_D^{\varepsilon}(\rho_{AB})$ is closely related to the conditional smooth max-entropy $H_{\max}^{\varepsilon}(A|B)_{\rho}$.
    The extracted randomness from $A$ in the presence of $E$, on the other hand, is linked to the conditional smooth min-entropy~$H_{\min}^{\varepsilon}(A|E)_{\rho}$. 
    These two entropies are dual to one another. That is, $H_{\min}^{\varepsilon}(A|E)=-H_{\max}^{\varepsilon}(A|B)$~\cite{tomamichel2010duality}. 

    One can now ask whether a similar duality relation holds when we switch to the so called quantum HILL entropies~\cite{chen2017computational}), which are the computational notion of smooth entropies.\footnote{In this case, computational efficiency is a requirement on the distinguisher; see~\cite{chen2017computational}.}
    Denoting the entropies in this computational setup by $\hat{H}_{\min}^{\varepsilon}(A|E)$ and $\hat{H}_{\max}^{\varepsilon}(A|B)$ (for the sake of this section), a duality relation such as $\hat{H}_{\min}^{\varepsilon}(A|E)=-\hat{H}_{\max}^{\varepsilon}(A|B)$ does not necessarily hold. The reason, again, lies in the complexity of switching between different purifications (to see this one should first consider the definitions of the smooth entropies and take into account that the entropies are defined via the purified distance). 

To summarize this section, we observe that our deep understanding, built over many years of explorations in quantum information theory, regarding the relations between entanglement, randomness and entropies does not simply translate to the computational setup: Many open questions are waiting to be answered at this new interface!

    \subsection{Strengthening the link to quantum gravity}\label{sec:ads_cft}

        Entanglement entropy plays an important role in the study of quantum gravity. 
        One well-known example is the Bekenstein-Hawking entropy that assigns an entropy to a black hole in terms of the area of its horizon~\cite{strominger1996microscopic}.  
        Another example is the AdS/CFT correspondence --- a duality between gravity in anti de Sitter space and conformal field theories on a lower dimension. In particular, the celebrated Ryu-Takayanagi formula~\cite{ryu2006holographic,ryu2006aspects} allows one to relate  entanglement entropy in the CFT to an area in AdS. 
        On the technical level, the entanglement entropy is a convenient quantity to work with in the context of quantum field theory as it can be related to path integrals via the so-called ``replica trick''. 
        In a different direction, once the connection was made between gravity and entanglement entropy, deeper connections between quantum gravity and quantum information theory  started to form. 

        Focusing on the topic most relevant for our work, novel ideas were developed to understand the \emph{computational complexity} of the AdS/CFT ``dictionary'' that specifies the exact correspondence between the physical theories~\cite{aaronson2016complexity,bouland2020computational}. 
        In particular, a couple of works~\cite{gheorghiu2020estimating,aaronson2022quantum} directly approached the question by focusing on entanglement entropy and studying the complexity of distinguishing states with low and high entanglement entropy. This is, arguably, also the main motivation for the definition of pseudo-entanglement in~\cite{aaronson2022quantum} via the entanglement entropy. 
        In both works~\cite{gheorghiu2020estimating,aaronson2022quantum} the goal in the context of quantum gravity was to expand the toolkit for constructing pseudo-entangled states with holographic entanglement structures for which the AdS/CFT correspondance is valid. (The works found necessary but not sufficient conditions.)

        Our new definition of pseudo-entanglement, Definition~\ref{def:ps-ent}, seems to capture holographic entanglement in a more refined way --- our definition opens the door for considering entangled states in CFT which, when ``translated'' to the bulk in AdS, admit a more complex geometric description.
        Capturing complex geometries seems to be necessary to conclude something reagarding the complexity of the AdS/CFT dictionary, as it is not expected to be hard to compute in all cases~\cite{engelhardt2021world}.
        The additional flexibility emerges by considering mixed states as well as from the way we define efficiency in terms of bipartite computation (Section~\ref{sec:bp_comp}) and the use of our computational entanglement measures (Section~\ref{sec:comp_ent_meas}). 
        Let us explain this in more detail.

        Different lines of work in high energy physics study holographic entanglement. 
        In~\cite{hayden2013holographic}, for example, it is shown that for a quantum state to have a geometric meaning in the bulk, certain conditions on the mutual information of a \emph{three-partite} system must hold. This means, in particular, that the bipartite marginals are not necessarily pure and one should consider mixed states as we do (with the purifying system being the third party).
        More examples for situations in which mixed states are relevant appear in~\cite{akers2019large,akers2021leading}.
        
        Furthermore, recent research suggests that, mathematically (that is, ignoring the ``concept'' of entanglement), the entanglement entropy is not always the right quantity to use in the AdS/CFT correspondence~\cite{akers2021leading,berkooz2022going,akers2023one}. Leading order corrections are needed when, e.g., one wishes to consider geometries with entangled black holes or mixed matter states~\cite{akers2021leading}.
        In~\cite{akers2023one} it was shown that the more relevant quantity is the \emph{conditional smooth max-entropy}, using ideas from one-shot information theory. 

        From the point of view of entanglement theory, this is not very surprising. As we mentioned throughout the previous sections, the entanglement entropy has no operational meaning as an entanglement measure once the quantum states are not pure. 
        The conditional smooth max-entropy, on the other hand, gives a lower bound on the  one-shot distillable entanglement $E_D^{\varepsilon}$ (recall Lemma~\ref{lem:cond_ent}).   
        This hints that in order to capture different geometries (as in~\cite{akers2021leading,akers2023one}) also in the computational case, working with the computational distillable entanglement $\hat{E}_D^{\varepsilon}$ may be more fruitful.

        The second aspect of our definition of pseudo-entanglement that differs from previous works and is of relevance for studying the complexity of the AdS/CFT correspondence is the definition of efficiency. 
        In our work we define efficiency in terms a efficient bipartite LOCC, with maximally entangled states as inputs/outputs. 
        This fits the notion of non-local computation, which is also used to describe computations that are being done on the boundary in the context of the AdS/CFT correspondence~\cite{may2019quantum,may2020holographic,may2022complexity}.  
        More specifically, one can consider a computation/task in the boundary CFT. As shown in~\cite{may2019quantum}, it should follow from the correspondence that tasks  that have an implementation in the bulk spacetime also have an implementation in the boundary CFT, while maintaining the success probability~\cite{may2021holographic}. Natural scattering tasks in the bulk translate into ``non-local computations'' in the CFT. Here a non-local computation is, loosely speaking, a computation that can be performed ``locally'' while using sufficiently many EPR pairs in order to achieve the required success probability. As shown in the literature some of these tasks may require very large number of EPR pairs. The question of how many EPR pairs can be distilled between two specific boundary regions (the ones associated with the task under consideration) thus becomes highly relevant. Moreover, when studying the computational efficiency of the bulk/boundary translation map it also becomes directly relevant to consider if the entanglement can be distilled efficiently.
        
				When discussing entanglement as a resource in CFT, it is thus expected that we should, firstly, do so while respecting a specific pre-defined bi-partition structure and, secondly, take efficiency into consideration. Both considerations affect the framework of non-local computation that arises so naturally in the AdS/CFT correspondence. They are material for the definition of pseudo-entanglement. 
    Indeed, as shown in Lemma~\ref{lem:no-comp-dist} and Lemma~\ref{lem:ec_gap} the fact that a state of low-entanglement entropy can be created efficiently without respecting the bipartition, as in~\cite{gheorghiu2020estimating,aaronson2022quantum}, does not, a priory at least, imply that it has low computational entanglement cost $\hat{E}_C^{\varepsilon}$, and even less low computational distillable entanglement $\hat{E}_D^{\varepsilon}$. Thus, results regarding pseudo-entangled states defined as in~\cite{gheorghiu2020estimating,aaronson2022quantum} do not necessarily translate to states with holographic entanglement structure, i.e., bulk with a geometric interpretation. 
         Our definition of pseudo-entanglement, on the other hand, fits the notion of efficiency used to study non-local computations in CFT, as discussed in Section~\ref{sec:bp_comp}, and thus may allow to strengthen the study of the complexity of the AdS/CFT dictionary.

\subsection*{Acknowledgments}
We thank Scott Aaronson, Adam Bouland and Chinmay Nirkhe for useful discussions regarding the work~\cite{aaronson2022quantum}, and Alex May and Simon Ross for referring to and explaining recent results about the AdS/CFT correspondence.
RAF was supported by the Israel Science Foundation (ISF), and the Directorate for Defense Research and Development (DDR\&D), grant No.\ 3426/21, by the Minerva foundation with funding from the Federal German Ministry for Education and Research, by the Peter and Patricia Gruber Award and by the Koshland Research Fund.
ZB was supported by the Israel Science Foundation (Grant No.\ 3426/21), and by the European Union Horizon 2020 Research and Innovation Program via ERC Project REACT (Grant 756482).
TV is supported by a research grant from the Center for New Scientists at the Weizmann Institute of Science, AFOSR Grant No. FA9550-22-1-0391, and ERC Consolidator Grant VerNisQDevS (101086733).

\bibliographystyle{alpha}
\bibliography{bib}

\end{document}